\documentclass[11pt]{article}
\usepackage{amsmath,amssymb,amsthm}
\usepackage{fullpage}
\usepackage{url}
\usepackage{color}

\DeclareMathOperator*{\E}{\mathbb{E}}
\let\Pr\relax
\DeclareMathOperator*{\Pr}{\mathbb{P}}

\newcommand{\eps}{\varepsilon}
\newcommand{\inprod}[1]{\left\langle #1 \right\rangle}
\newcommand{\R}{\mathbb{R}}
\newcommand{\sr}{\tilde{r}}
\newcommand{\nr}{\tilde{nr}}
\newcommand{\rank}{r}
\newcommand{\Tr}{\mathop{tr}}

\newtheorem{theorem}{Theorem}
\newtheorem{remark}{Remark}
\newtheorem{lemma}{Lemma}
\newtheorem{corollary}{Corollary}
\newtheorem{definition}{Definition}

\newcommand{\proofbelow}{3pt}
\newcommand{\afterproof}{\hfill $\blacksquare$ \par \vspace{\proofbelow}}
\renewenvironment{proof}{\noindent\textbf{Proof.}\,}{\afterproof}

\newcommand{\EquationName}[1]{\label{eq:#1}}

\newcommand{\LemmaName}[1]{\label{lem:#1}}
\newcommand{\CorollaryName}[1]{\label{cor:#1}}
\newcommand{\SectionName}[1]{\label{sec:#1}}
\newcommand{\TheoremName}[1]{\label{thm:#1}}
\newcommand{\RemarkName}[1]{\label{rem:#1}}

\newcommand{\Equation}[1]{Eq.\:\eqref{eq:#1}}

\newcommand{\Lemma}[1]{Lemma~\ref{lem:#1}}
\newcommand{\Corollary}[1]{Corollary~\ref{cor:#1}}
\newcommand{\Section}[1]{Section~\ref{sec:#1}}
\newcommand{\Theorem}[1]{Theorem~\ref{thm:#1}}
\newcommand{\Remark}[1]{Remark~\ref{rem:#1}}

\author{Michael B.\ Cohen\thanks{MIT. \texttt{micohen@mit.edu}. Supported by Akamai Presidential Fellowship and NSF grant CCF-1111109.}\and Jelani Nelson\thanks{Harvard. \texttt{minilek@seas.harvard.edu}. Supported by NSF grant IIS-1447471 and CAREER CCF-1350670, ONR grant N00014-14-1-0632 and Young Investigator N00014-15-1-2388, and a Google Faculty Research Award. }\and David P.\ Woodruff\thanks{IBM Almaden. \texttt{dpwoodru@us.ibm.com}.
Supported by XDATA program  of the Defense Advanced Research Projects Agency (DARPA), administered through Air Force Research Laboratory FA8750-12-C-0323.}}

\title{Optimal approximate matrix product in terms of stable rank}

\date{}

\begin{document}

\setcounter{page}{0}

\maketitle

\thispagestyle{empty}

\begin{abstract}
We give two different characterizations of the type of dimensionality-reducing map $\Pi$ that can be used for spectral error approximate matrix multiplication (AMM). Both imply a random data-oblivious $\Pi$ with $m = O(\sr/\eps^2)$ rows suffices, where $\sr$ is the maximum {\em stable rank}, i.e.\ squared ratio of Frobenius and operator norms, of the matrices being multiplied. This answers the main open question of \cite{MagenZ11,KyrillidisVZ14}, and is optimal for any random oblivious map. 

Both characterizations apply to a general class of random $\Pi$, and one even to deterministic $\Pi$. Recall an $(\eps,\delta,d)$-{\em oblivious subspace embedding (OSE)} distribution $\mathcal{D}$ over matrices $\Pi\in\R^{m\times n}$ is such that for any $d$-dimensional linear subspace $E$ of $\R^n$, $\Pr(\|(\Pi U)^T (\Pi U) - I\| > \eps) < \delta$, where the columns of $U$ form an orthonormal basis for $E$. In one characterization, we show if this tail bound was established via the moment method, then to obtain AMM it suffices that $\mathcal{D}$ be an $(\eps, \delta, 2\sr)$-OSE. That is, we show being an OSE for dimension (i.e.\ {\em rank}) $k$ implies {\em black box} AMM for matrices of {\em stable rank} $k$. Once this is shown, our main result is then just a simple corollary of the fact that subgaussian maps with $m = \Omega((d+\log(1/\delta))/\eps^2)$ rows are $(\eps, \delta, d)$-OSE's. Also, for all known OSE's, the best analyses indeed are via the moment method (or tools such as matrix Chernoff, which themselves also imply moment bounds).  Thus our theorem can be applied to a much more general class of sketching matrices than just the subgaussian sketches in \cite{MagenZ11,KyrillidisVZ14}, in addition to achieving better bounds. This includes fast subspace embeddings such as the Subsampled Randomized Hadamard Transform \cite{Sarlos06,LibertyWMRT07} or sparse subspace embeddings \cite{ClarksonW13,MengM13,NelsonN13,Cohen16}, or even constructions that may be developed in the future, to show that rank bounds in their analyses in previous work can automatically be replaced with stable rank. Our second characterization identifies certain deterministic conditions which if satisfied imply the AMM guarantee. We show these conditions are sufficiently precise to yield optimal results for subgaussian maps and even a deterministic $\Pi$ such as the truncated SVD.

Our main theorem, via connections with spectral error matrix multiplication proven in previous work, implies quantitative improvements for approximate least squares regression and low rank approximation \cite{Sarlos06}, and implies faster low rank approximation for popular kernels in machine learning such as the gaussian and Sobolev kernels. Our main result has also already been applied to improve dimensionality reduction guarantees for $k$-means clustering \cite{CohenEMMP14}, and also implies new results for nonparametric regression when combined with results in \cite{YangPW15}.

Lastly, we point out a minor but interesting observation that the proof of the ``BSS'' deterministic row-sampling result of \cite{BatsonSS12} can be modified to show that for any matrices $A, B$ of stable rank at most $\sr$, one can achieve the spectral norm guarantee for approximate matrix multiplication of $A^T B$ using a deterministic sampling matrix with $O(\sr/\eps^2)$ non-zero entries which can be found in polynomial time. The original result of \cite{BatsonSS12} was for rank instead of stable rank. Our observation leads to a stronger version of a main theorem of \cite{KollaMST10}.
\end{abstract}

\section{Introduction}\SectionName{intro}
Much recent work has successfully utilized randomized dimensionality reduction techniques to speed up solutions to linear algebra problems, with applications in machine learning, statistics, optimization, and several other domains; see the recent monographs \cite{HMT11,Mahoney11,Woodruff14} for more details. In our work here, we give new spectral norm guarantees for approximate matrix multiplication (AMM). Aside from AMM being interesting in its own right, it has become a useful primitive in the literature for analyzing algorithms for other large-scale linear algebra problems as well. We show applications of our new guarantees to speeding up standard algorithms for generalized regression and low-rank approximation problems. We also describe applications of our results to $k$-means clustering (discovered in \cite{CohenEMMP14}) and nonparametric regression \cite{YangPW15}. 

In AMM we are given $A, B$ each with a large number of rows $n$, and the goal is to compute some matrix $C$ such that $\|C - A^T B\|_X$ is ``small'', for some norm $\|\cdot\|_X$. Furthermore, we would like to compute $C$ much faster than the usual time required to exactly compute $A^T B$.

Work on randomized methods for AMM began with \cite{DrineasKM06}, which focused on $\|\cdot\|_X = \|\cdot\|_F$, i.e.,\ Frobenius norm error. They showed by picking an appropriate sampling matrix $\Pi\in\R^{m\times n}$, 
\begin{equation}
\|(\Pi A)^T(\Pi B) - A^T B\|_F \le \eps\|A\|_F\|B\|_F \EquationName{frob-error}
\end{equation}
with good probability, if $m = \Omega(1/\eps^2)$. By a {\em sampling matrix}, we mean the rows of $\Pi$ are independent, and each row is all zero except for a $1$ in a random location according to some appropriate distribution. If $A\in\R^{n\times d}$ and $B\in\R^{n\times p}$, note $(\Pi A)^T (\Pi B)$ can be computed in $O(mdp)$ time once $\Pi A$ and $\Pi B$ are formed, as opposed to the straightforward $O(ndp)$ time to compute $A^T B$.

The Frobenius norm error guarantee of \Equation{frob-error} was also later achieved in \cite[Lemma 6]{Sarlos06} via a different approach, with some later optimizations to the parameters in \cite[Theorem 6.2]{KaneN14}. The approach of Sarl\'{o}s was not via sampling, but rather to use a matrix $\Pi$ drawn from a distribution satisfying an ``oblivious Johnson-Lindenstrauss (JL)'' guarantee, i.e.\ a distribution $\mathcal{D}$ over $\R^{m\times n}$ satisfying the following condition for some $\eps, \delta\in(0, 1/2)$:
\begin{equation}
\forall x\in\R^n,\ \Pr_{\Pi\sim\mathcal{D}}\left(|\|\Pi x\|_2^2 - \|x\|_2^2| > \eps \|x\|_2^2\right) < \delta .
\end{equation}
Such a matrix $\Pi$ can be taken with $m = O(\eps^{-2}\log(1/\delta))$ \cite{JL84}. Furthermore, one can take $\Pi$ to be a Fast JL transform \cite{AilonC09} (or any of the follow-up improvements \cite{AilonL13,KrahmerW11,NelsonPW14,Bourgain14,HavivR16}) or a sparse JL transform \cite{dks10,KaneN14} to speed up the computation of $\Pi A$ and $\Pi B$. One could also use the Thorup-Zhang sketch \cite{ThorupZ12} combined with a certain technique of \cite{LiangBKW14} (see \cite[Theorem 2.10]{Woodruff14} for details) to efficiently boost success probability.

Other than Frobenius norm error, the main other error guarantee investigated in previous work is spectral error. That is, we would like $\|C - A^T B\|$ to be small, where $\|M\|$ denotes the largest singular value of $M$. If one is interested in applying $A^T B$ to some set of input vectors then this type of error is the most meaningful, since $\|C - A^T B\|$ being small is equivalent to $\|Cx\| \approx \|A^T B x\|$ for any $x$. The first work along these lines was again by \cite{DrineasKM06}, who gave a procedure based on entry-wise sampling of the entries of $A$ and $B$. The works \cite{DrineasMM06,SpielmanS11} showed that row-sampling according to leverage scores also provides the desired guarantee with few samples.

Then \cite{Sarlos06}, combined with a quantitative improvement in \cite{ClarksonW13}, showed that one can take a $\Pi$ drawn from an oblivious JL distribution with $\delta = 2^{-\Theta(r)}$ where $\rank(\cdot)$ denotes rank and $\rank = \rank(A) + \rank(B)$. Then for $\Pi$ with $m = O((\rank+\log(1/\delta))/\eps^2)$, with probability at least $1-\delta$ over $\Pi$,
\begin{equation}
\|(\Pi A)^T (\Pi B) - A^T B\| \le \eps \|A\|\|B\| . \EquationName{op-error}
\end{equation}
As we shall see shortly via a very simple lemma (\Lemma{simple-mmult}), a sufficient deterministic condition implying \Equation{op-error} is that $\Pi$ is an $O(\eps)$-{\em subspace embedding} for the $r$-dimensional subspace spanned by the columns of $A, B$. The notion of a subspace embedding was introduced by \cite{Sarlos06}.

\begin{definition}
$\Pi$ is an {\em $\eps$-subspace embedding} for $U\in\R^{n\times \rank}$, $U^T U = I$, if $\Pi$ satisfies \Equation{op-error} with $A = B = U$, i.e. $\|(\Pi U)^T (\Pi U) - I\| \le \eps$. This is equivalent to $\forall x\in \R^{\rank},\ (1-\eps)\|x\|_2^2 \le \|\Pi Ux\|_2^2 \le (1+\eps)\|x\|_2^2$, i.e.\ $\Pi$ preserves norms of all vectors in the subspace spanned by the columns $U$.

An {\em $(\eps, \delta, \rank)$-oblivious subspace embedding (OSE)} is a distribution $\mathcal{D}$ over $\R^{m\times n}$ such that 
$$
\forall U\in\R^{n\times\rank},\ U^TU = I,\ \Pr_{\Pi\sim \mathcal{D}}(\|(\Pi U)^T (\Pi U) - I\| > \eps) < \delta .
$$
\end{definition}

Fast subspace embeddings $\Pi$, i.e.\ such that the products $\Pi A$ and $\Pi B$ can be computed quickly, are known using variants on the Fast JL transform such as the Subsampled Randomized Hadamard Transform (SRHT) \cite{Sarlos06,LibertyWMRT07,Tropp11,LuDFU13} (also see a slightly improved analysis of the SRHT in \Section{srht}) or via sparse subspace embeddings \cite{ClarksonW13,MengM13,NelsonN13,LiMP13,CohenLMMPS15,Cohen16}. In most applications it is important to have a fast subspace embedding to shrink the time it takes to transform the input data to a lower-dimensional form. The SRHT is a construction of a $\Pi$ such that $\Pi A$ can be computed in time $O(nd\log n)$ (see \Section{srht} for details of the construction). The sparse subspace embedding constructions have some parameter $m$ rows and exactly $s$ non-zero entries per column, so that $\Pi A$ can be computed in time $O(s\cdot \mathop{nnz}(A))$, where $\mathop{nnz}(\cdot)$ is the number of non-zero entries, and there is a tradeoff in the upper bounds between $m$ and $s$.

An issue addressed by the work of \cite{MagenZ11} is that of robustness. As stated above, achieving \Equation{op-error} requires $\Pi$ be a subspace embedding for an $r$-dimensional subspace. However, consider the case when $A$ (and similarly for $B$) is of high rank but can be expressed as the sum of a low-rank matrix plus high-rank noise of small magnitude, i.e.,\ $A = \tilde{A} + E_A$ for $\tilde{A}$ of rank $\rank(\tilde{A}) \ll \rank$, and where $\|E_A\|$ is very small but $E_A$ has high (even full) rank. One would hope the noise could be ignored, but standard results require $\Pi$ to have a number of rows at least as large as $\rank$, regardless of how small the magnitude of the noise is. Another case of interest (as we will see in \Section{applications}) is when $A$ and $B$ are each of high rank, but their singular values decay at some appropriate rate. As discussed in \Section{applications}, in several applications where AMM is not the final goal but rather is used as a primitive in analyzing an algorithm for some other problem (such as $k$-means clustering or nonparametric regression), the matrices that arise do indeed have such decaying singular values.

The work \cite{MagenZ11} remedied this by considering the {\em stable ranks} $\sr(A), \sr(B)$ of $A$ and $B$. Define $\sr(A) = \|A\|_F^2 / \|A\|^2$. Note $\sr(A) \le \rank(A)$ always, but can be much less if $A$ has a small tail of singular values. Let $\sr$ denote $\sr(A) + \sr(B)$. Among other results, \cite{MagenZ11} showed that to achieve \Equation{op-error} with good probability, one can take $\Pi$ to be a random (scaled) sign matrix with either $m = \Omega(\sr / \eps^4)$ or $m = \Omega(\sr\log(d+p)/\eps^2)$ rows. As noted in follow-up work \cite{KyrillidisVZ14}, both the $1/\eps^4$ dependence and the $\log(d+p)$ factor are undesirable. In their data-driven low dimensional embedding application, they wanted a dimension $m$ independent of the original dimensions, which are assumed much larger than the stable rank, and also wanted lower dependence on $1/\eps$. To this end, \cite{KyrillidisVZ14} defined the {\em nuclear rank} as $\nr(A) = \|A\|_* / \|A\|$ and showed $m = \Omega(\nr/\eps^2)$ rows suffice for $\nr = \nr(A) + \nr(B)$. Here $\|A\|_*$ is the nuclear norm, i.e.,\ sum of singular values of $A$. Since $\|A\|_F^2$ is the sum of squared singular values, it is straightforward to see that $\nr(A) \ge \sr(A)$ always. Thus there is a tradeoff: the stable rank guarantee is worsened to nuclear rank, but dependence on $1/\eps$ is improved to quadratic.

We show switching to the weaker $\nr$ guarantee is unnecessary by showing quadratic dependence on $1/\eps$ holds even with stable rank. This answers the main open question of \cite{MagenZ11,KyrillidisVZ14}. 

To state our results in a more natural way, we rephrase our main result to say that we achieve 
\begin{equation}
\|(\Pi A)^T (\Pi B) - A^T B\| \le \eps \sqrt{\left(\|A\|^2 + \frac{\|A\|_F^2}k\right)\left(\|B\|^2 + \frac{\|B\|_F^2}k\right)}. \EquationName{mmult-general}
\end{equation}
for an arbitrary $k\ge 1$, and we do so by using subspace embeddings for $O(k)$-dimensional subspaces in a certain black box way (which will be made precise soon) regardless of the ranks of $A, B$. 

\begin{remark}\RemarkName{equiv}
\textup{
Note that our previously stated main contribution is equivalent, since one could set $k = \sr(A) + \sr(B)$ to arrive at the conclusion that subspace embeddings for $O(\sr)$-dimensional subspaces yield the guarantee in \Equation{op-error}. Alternatively one could obtain the \Equation{mmult-general} guarantee via \Equation{op-error} with error parameter $\eps' = \Theta(\eps\cdot \min\{1, \sqrt{(\sr(A)\cdot \sr(B))/k}\})$. 
}
\end{remark}

Henceforth, we use the following definition.

\begin{definition}
For conforming matrices $A^T, B$, we say $\Pi$ satisfies the {\em $(k, \eps)$-approximate spectral norm matrix multiplication property ($(k,\eps)$-AMM) for $A, B$} if \Equation{mmult-general} holds. If $\Pi$ is random and satisfies $(k,\eps)$-AMM with probability $1-\delta$ for any fixed $A, B$, then we say $\Pi$ satisfies $(k,\eps,\delta)$-AMM.
\end{definition}

\paragraph{Our main contribution:} We give two different characterizations for $\Pi$ supporting $(k,\eps)$-AMM, both of which imply $(k,\eps,\delta)$-AMM $\Pi$ having $m = O((k+\log(1/\delta))/\eps^2)$ rows. The first characterization applies to any OSE distribution for which a moment bound has been proven for $\|(\Pi U)^T(\Pi U) - I\|$ (which is true for the best analyses of all known OSE's). In this case, we show a black box theorem: any $(\eps, \delta, 2k$)-OSE provides $(k,\eps,\delta)$-AMM. Since matrices with subgaussian entries and $m = \Omega((k+\log(1/\delta))/\eps^2)$ are $(\eps, \delta, 2k)$-OSE's, our originally stated main result follows. This result is optimal, since \cite{nn14} shows any randomized distribution over $\Pi$ with $m$ rows having the $(k,\eps, \delta)$-AMM property must have $m = \Omega((k + \log(1/\delta))/\eps^2)$ (the hard instance there is when $A = B = U$ has orthonormal columns, and thus rank and stable rank are equal).

Our second characterization identifies certain deterministic conditions which, if satisfied by $\Pi$, imply the desired $(k,\eps)$-AMM property. These conditions are of the form: (1) $\Pi$ should preserve a certain set of $O(\log(1/\eps))$ different subspaces of varying dimensions (all depending on $k, \eps$ and not on the ranks of $A, B$) with varying distortions, and (2) for a certain two matrices in our analysis, left-multiplication by $\Pi$ should not increase their operator norms by more than an $O(1)$ factor. These conditions are chosen carefully so that matrices with subgaussian entries and $m = \Omega(k/\eps^2)$ satisfy all conditions simultaneously with high probability, again thus proving our main result while also suggesting that the conditions we have identified are the ``right'' ones.

Due to the black box reliance on the subspace embedding primitive in our proofs, $\Pi$ need not only be a subgaussian map. Thus not only do we improve on $m$ compared with previous work, but also in terms of the general class of $\Pi$ our result applies to. For example given our first characterization, not only does it suffice to use a random sign matrix with $\Omega(k/\eps^2)$ rows, but in fact one can apply our theorem to more efficient subspace embeddings such as the SRHT or sparse subspace embeddings, or even constructions discovered in the future. That is, one can automatically transfer bounds proven for the subspace embedding property to the $(k,\eps)$-AMM property. Thus, for example, the best known SRHT analysis (in our appendix, see \Theorem{best-srht}) implies $(k,\eps, \delta)$-AMM for $m = \Omega((k + \log(1/(\eps\delta))\log(k/\delta))/\eps^2)$ rows. For sparse subspace embeddings, the analysis in \cite{Cohen16} implies $m = \Omega(k\log(k/\delta)/\eps^2)$ suffices with $s = O(\log(k/\delta)/\eps)$ non-zeroes per column of $\Pi$. The only reason for the $\log k$ loss in $m$ for these particular distributions is not due to our theorems, but rather due to the best analyses for the simpler {\em subspace embedding} property in previous work already incurring the extra $\log k$ factor (note being a subspace embedding for a $k$-dimensional subspace is simply a special case of $(k,\eps)$-AMM where $A = B = U$ has $k$ orthonormal columns). In the case of the SRHT, this extra $\log k$ factor is actually necessary \cite{Tropp11}; for sparse subspace embeddings, it is conjectured that the $\log k$ factor can be removed and that $m = \Omega((k+\log(1/\delta))/\eps^2)$ actually suffices to obtain an OSE \cite[Conjecture 14]{NelsonN13}. We also discuss in \Remark{efficiency} that one can set $\Pi$ to be $\Pi_1\cdot \Pi_2$ where $\Pi_1$ has subgaussian entries with $O(k/\eps^2)$ rows, and $\Pi_2$ is some other fast OSE (such as the SRHT or sparse subspace embedding), and thus one could obtain the best of both worlds: (1) $\Pi$ has $O(k/\eps^2)$ rows, and (2) can be applied to any $A\in\R^{n\times d}$ in time $T + O(km'd/\eps^2)$, where $T$ is the (fast) time to apply $\Pi_2$ to $A$, and $m'$ is the number of rows of $\Pi_2$. For example, by appropriate composition as discussed in \Remark{efficiency}, $\Pi$ can have $O(k/\eps^2)$ rows and support multiplying $\Pi A$ for $A\in\R^{n\times d}$ in time $O(\mathop{nnz(A)}) + \tilde{O}(\eps^{-O(1)}(k^3 + k^2d))$.

We also observe the proof of the main result of \cite{BatsonSS12} can be modified to show that given any $A, B$ each with $n$ rows, and given any $\eps\in(0,1/2)$, there exists a diagonal matrix $\Pi\in\R^{n\times n}$ with $O(k/\eps^2)$ non-zero entries, and that can be computed by a deterministic polynomial time algorithm, achieving $(k,\eps)$-AMM. The original work of \cite{BatsonSS12} achieved \Equation{op-error} with $m = O(\rank/\eps^2)$ for $\rank$ being the sum of ranks of $A, B$. The work \cite{BatsonSS12} stated their result for the case $A = B$, but the general case of potentially unequal matrices reduces to this case; see \Section{bss}. Our observation also turns out to yield a stronger form of \cite[Theorem 3.3]{KollaMST10}; also see \Section{bss}.

\bigskip

As mentioned, aside from AMM being interesting on its own, it is a useful primitive widely used in analyses of algorithms for several other problems, including $k$-means clustering \cite{BoutsidisZMD15,CohenEMMP14}, nonparametric regression \cite{YangPW15}, linear least squares regression and low-rank approximation \cite{Sarlos06}, approximating leverage scores \cite{DrineasMMW12}, and several other problems (see \cite{Woodruff14} for a recent summary).  For all these, analyses of correctness for algorithms based on dimensionality reduction via some $\Pi$ rely on $\Pi$ satisfying AMM for certain matrices in the analysis.

After making certain quantitative improvements to connections between AMM and applications, and combining them with our main result, in \Section{applications} we obtain the following new results.

\begin{enumerate}
\item \textbf{Generalized regression:} Given $A\in\R^{n\times d}$ and $B\in\R^{n\times p}$, consider the problem of computing $X^* = \mathop{argmin}_{X\in\R^{d\times p}} \|AX - B\|$. It is standard that $X^* = (A^T A)^+ A^T B$ where $(\cdot)^+$ is the Moore-Penrose pseudoinverse. The bottleneck here is computing $A^T A$, taking $O(nd^2)$ time. A popular approach is to instead compute $\tilde{X} = ((\Pi A)^T(\Pi A))^+ (\Pi A)^T \Pi B$, i.e.,\ the minimizer of $\|\Pi A X - \Pi B\|$. Note that computing $(\Pi A)^T (\Pi A)$ (given $\Pi A$) only takes a smaller $O(md^2)$ amount of time. We show that if $\Pi$ satisfies $(k, O(\sqrt{\eps}))$-AMM for $U_A, P_{\bar{A}} B$, and is also an $O(1)$-subspace embedding for a certain $\rank(A)$-dimensional subspace (see \Theorem{gen-regression}), then
$$
\| A \tilde X - B \|^2 \leq (1 + \eps) \| P_A B - B \|^2 + (\eps / k) \| P_A B - B \|_F^2
$$
where $P_A$ is the orthogonal projection onto the column space of $A$, $P_{\bar{A}} = I - P_A$, and $U_A$ has orthonormal columns forming a basis for the column space of $A$. The punchline is that if the regression error $P_{\bar{A}} B$ has high actual rank but stable rank only on the order of $\rank(A)$, then we obtain multiplicative spectral norm error with $\Pi$ having fewer rows. Generalized regression is a natural extension
of the case when $B$ is a vector, 
and arises for example 
in Regularized Least Squares Classification, where one has multiple (non-binary) labels, and for
each label one creates a column of $B$; see e.g.\ \cite{CLLLH10} for this and variations. 

\item \textbf{Low-rank approximation:} We are given $A\in\R^{n\times d}$ and integer $k\ge 1$, and we want to compute $A_k = \mathop{argmin}_{\rank(X) \le k} \|A - X\|$. The Eckart-Young theorem implies $A_k$ is obtained by truncating the SVD of $A$ to the top $k$ singular vectors. The standard way to use dimensionality reduction for speedup, introduced in \cite{Sarlos06}, is to let $S = \Pi A$ then compute $\tilde{A} = A P_S$. Then return $\tilde{A}_k$, the best rank-$k$ approximation of $\tilde{A}$, instead of $A_k$ (it is known $\tilde{A}_k$ can be computed more efficiently than $A_k$; see \cite[Lemma 4.3]{ClarksonW09}). We show if $\Pi$ satisfies $(k, O(\sqrt{\eps}))$-AMM for $U_k$ and $A-A_k$, and is a $(1/2)$-subspace embedding for the column space of $A_k$, then
$$
\|\tilde{A}_k - A\|^2 \le (1+\eps)\|A - A_k\|^2 + (\eps/k)\|A - A_k\|_F^2 .
$$
The punchline is that if the stable rank of the tail $A - A_k$ is on the same order as the rank parameter $k$, then standard algorithms from previous work for Frobenius multiplicative error actually in fact also provide {\em spectral} multiplicative error. This property indeed holds for any $k$ for popular kernel matrices in machine learning such as the gaussian and Sobolev kernels (see \cite{ReyhaniHV11} and Examples 2 and 3 of \cite{YangPW15}), and low-rank approximation of kernel matrices has been applied to several machine learning problems; see \cite{GittensM13} for a discussion.
\end{enumerate}

We also explain in \Section{applications} how our result has already been applied in recent work on dimensionality reduction for $k$-means clustering \cite{CohenLMMPS15}, and how it generalizes results in \cite{YangPW15} on dimensionality reduction for nonparametric regression to use a larger class of embeddings $\Pi$.

\subsection{Preliminaries and notation}

We frequently use the singular value decomposition (SVD). For a matrix $A\in\R^{n\times d}$ of rank $\rank$, consider the compact SVD $A = U_A \Sigma_A V_A^T$ where $U_A \in \R^{n\times r}$ and $V_A\in\R^{d\times r}$ each have orthonormal columns, and $\Sigma_A$ is diagonal with strictly positive diagonal entries (the singular values of $A$). We assume $(\Sigma_A)_{i,i} \ge (\Sigma_A)_{j,j}$ for $i< j$. We let $P_A = U_AU_A^T$ denote the orthogonal projection operator onto the column space of $A$. We use $\mathrm{span}(A)$ to refer to the subspace spanned by $A$'s columns.

Often for a matrix $A$ we write $A_k$ as the best rank-$k$ approximation to $A$ under Frobenius or spectral error (obtained by writing the SVD of $A$ then setting all $(\Sigma_A)_{i,i}$ to $0$ for $i>k$). We often denote $A - A_k$ as $A_{\bar{k}}$. For matrices with orthonormal columns, such as $U_A$, $(U_A)_k$ denotes the $n\times k$ matrix formed by removing all but the first $k$ columns of $U$. When $A$ is understood from context, we often write $U \Sigma V^T$ instead of $U_A \Sigma_A V_A^T$, and $U_k$ to denote $(U_A)_k$ (and $\Sigma_k$ for $(\Sigma_A)_k$, etc.).

\section{Analysis of matrix multiplication for stable rank}\SectionName{analysis}

First we record a simple lemma relating subspace embeddings and AMM.

\begin{lemma}\LemmaName{simple-mmult}
Let $E = \mathrm{span}\{A, B\}$, and let $\Pi$ be an $\eps$-subspace embedding for $E$. Then \Equation{op-error} holds.
\end{lemma}
\begin{proof}
First, without loss of generality we may assume $\|A\| = \|B\| = 1$ since we can divide both sides of
\Equation{op-error} by $\|A\| \cdot \|B\|$. Let $U$ be a matrix whose columns form an orthonormal basis for $E$. Then note for any $x, y$ we can write $Ax = Uw, By = Uz$ where $\|w\| \le \|x\|, \|z\| \le \|y\|$. Then
\allowdisplaybreaks
\begin{align*}
\|(\Pi A)^T (\Pi B) - A^T B\| &= \sup_{\|x\|=\|y\|=1} |\inprod{\Pi A x, \Pi By} - \inprod{Ax, By}|\\
{}&= \sup_{\|w\|, \|z\| \le 1} |\inprod{\Pi U z, \Pi U w} - \inprod{Uz, Uw}|\\
{}&= \|(\Pi U)^T (\Pi U) - I\|\\
{}& < \eps
\end{align*}
\end{proof}

\Lemma{simple-mmult} implies that if $A, B$ each have rank at most $r$, it suffices for $\Pi$ to have $\Omega(r/\eps^2)$ rows.

In the following two subsections, we give two different characterizations for $\Pi$ to provide $(k,\eps)$-AMM, both only requiring $\Pi$ to have $\Omega(k/\eps^2)$ rows, independent of $r$.

\subsection{Characterization for $(k,\eps,\delta)$-AMM via a moment property}\SectionName{moments}

Here we provide a way to obtain $(k,\eps)$-AMM for any $\Pi$ whose subspace embedding property has been established using the moment method, e.g.\ sparse subspace embeddings \cite{MengM13,NelsonN13,Cohen16}, dense subgaussian matrices as analyzed in \Section{subgaussian-ose-moments}, or even the SRHT as analyzed in \Section{srht}. Our approach in this subsection is inspired by the introduction of the ``JL-moment property'' in \cite{KaneN14} to analyze approximate matrix multiplication with Frobenius error. The following is a generalization of \cite[Definition 6.1]{KaneN14}, which was only concerned with $d=1$.

\begin{definition}
A distribution $\mathcal{D}$ over $\R^{m\times n}$ has {\em $(\eps,\delta,d,\ell)$-OSE moments} if for all matrices $U\in\R^{n\times d}$ with orthonormal columns,
$$
\E_{\Pi\sim\mathcal{D}} \left\|(\Pi U)^T(\Pi U) - I\right\|^\ell < \eps^\ell \cdot \delta
$$
\end{definition}
Note that this is just a special case of bounding the expectation of an arbitrary function of $\|(\Pi U)^T(\Pi U) - I\|$.  The arguments below will actually apply to any nonnegative, convex, increasing function of $\|(\Pi U)^T(\Pi U) - I\|^2$, but we restrict to moments for simplicity of presentation. The acronym ``OSE'' refers to {\em oblivious subspace embedding}, a term coined in \cite{NelsonN13} to refer to distributions over $\Pi$ yielding a subspace embedding for any fixed subspace of a particular bounded dimension with high probability. We start with a simple lemma.

\begin{lemma}\LemmaName{matmult}
Suppose $\mathcal{D}$ satisfies the $(\eps,\delta,2d,\ell)$-OSE moment property and $A, B$ are matrices with (1) the same number of rows, and (2) sum of ranks at most $2d$. Then
$$
\E_{\Pi\sim\mathcal{D}} \left\|(\Pi A)^T(\Pi B) - A^T B\right\|^\ell < \eps^\ell \|A\|^\ell \|B\|^\ell \cdot \delta
$$
\end{lemma}
\begin{proof}
First, we apply \Lemma{simple-mmult} to $A$ and $B$, where $U$ forms an orthonormal basis for the subspace $\mathrm{span}\{\mathrm{columns}(A),\mathrm{columns}(B)\}$, showing that
$$
\left\|(\Pi A)^T(\Pi B) - A^T B\right\| \le \left\|(\Pi U)^T(\Pi U) - I\right\| \|A\| \|B\| .
$$
Therefore
$$
\E_{\Pi\sim\mathcal{D}} \left\|(\Pi A)^T(\Pi B) - A^T B\right\|^\ell \le \E_{\Pi\sim\mathcal{D}} \left\|(\Pi U)^T(\Pi U) - I\right\|^\ell \|A\|^\ell \|B\|^\ell < \eps^\ell \|A\|^\ell \|B\|^\ell \cdot \delta
$$
\end{proof}

Then, just as \cite[Theorem 6.2]{KaneN14} showed that having OSE moments with $d=1$ implies approximate matrix multiplication with Frobenius norm error, here we show that having OSE moments for larger $d$ implies approximate matrix multiplication with operator norm error. 

\begin{theorem}\TheoremName{main2}
Given $k,\eps,\delta \in (0, 1/2)$, let $\mathcal{D}$ be any distribution over matrices with $n$ columns with the $(\eps,\delta,2k,\ell)$-OSE moment property for some $\ell\ge 2$. Then, for any $A,B$,
\begin{equation}
\Pr_{\Pi\sim\mathcal{D}} \left ( \|(\Pi A)^T(\Pi B) - A^T B\| > \eps \sqrt{(\|A\|^2 + \|A\|_F^2 / k) (\|B\|^2 + \|B\|^2_F / k)} \right ) < \delta \EquationName{moment-thm}
\end{equation}
\end{theorem}
\begin{proof}
We can assume $A, B$ each have orthogonal columns. This is since, via the full SVD, there exist orthogonal matrices $R_A, R_B$ such that $A R_A$ and $B R_B$ each have orthogonal columns. Since neither left nor right multiplication by an orthogonal matrix changes operator norm, 
$$
\|(\Pi A)^T(\Pi B) - A^T B\| = \|(\Pi A R_A)^T(\Pi B R_B) - (A R_A)^T B R_B\| .
$$

Thus, we replace $A$ by $A R_A$ and similarly for $B$. We may also assume the columns $a_1, a_2, \ldots$ of $A$ are sorted so that $\|a_i\|_2 \ge \|a_{i+1}\|_2$ for all $i$. Henceforth we assume $A$ has orthogonal columns in this sorted order (and similarly for $B$, with columns $b_i$). Now, treat $A$ as a block matrix in which the columns are blocked into groups of size $k$, and similarly for $B$ (if the number of columns of either $A$ or $B$ is not divisible by $k$, then pad them with all-zero columns until they are, which does not affect the claim). Let the spectral norm of the $i$th block of $A$ be $s_i = \|a_{(i-1)\cdot k + 1}\|_2$, and for $B$ denote the spectral norm of the $i$th block as $t_i = \|b_{(i-1)\cdot k + 1}\|_2$. These equalities for $A, B$ hold since their columns are orthogonal and sorted by norm. We claim $\sum_i s_i^2 \le \|A\|^2 + \|A\|_F^2 / k$ (and similarly for $\sum_i t_i^2$). To see this, let the blocks of $A$ be $A'_1, \ldots, A'_q$ where $s_i = \|A'_i\|$. Note $s_1^2 = \|A'_1\| \le \|A\|$. Also, for $i>1$ we have
$$
s_i^2 = \|a_{(i-1)\cdot k + 1}\|_2^2\le \frac 1{k} \sum_{(i-2)\cdot k + 1 \le j \le (i-1)\cdot k} \|a_j\|_2^2 = \frac 1{k} \|A'_{i-1}\|_F^2 .
$$
Thus
$$
\sum_{i>1} s_i^2 \le \|A\|_F^2 / k.
$$

Define $C = (\Pi A)^T(\Pi B) - A^T B$.  Let $v_{\{ i \}}$ denote the $i$th block of a vector $v$ (the $k$-dimensional vector whose entries consist of entries $(i-1)\cdot k + 1$ to $i\cdot k$ of $v$), and $C_{\{i\},\{j\}}$ the $(i,j)$th block of $C$, a $k\times k$ matrix (the entries in $C$ contained in the $i$th block of rows and $j$th block of columns).

Now, $\|C\| = \sup_{\|x\|=\|y\|=1} x^T C y$.  For any such vectors $x$ and $y$, we define new vectors $x'$ and $y'$ whose coordinates correspond to entire blocks: we let $x'_i = \| x_{\{ i \}} \|$, with $y'$ defined analogously.  We similarly define $C'$ with entries corresponding to blocks of $C$, where $C'_{i,j} = \| C_{\{i\},\{j\}} \|$. Then $x^T C y \le x'^T C' y'$, simply by bounding the contribution of each block. Thus it suffices to upper bound $\|C'\|$, which we bound by its Frobenius norm $\|C'\|_F$. Now, recalling for a random variable $X$ that $\|X\|_\ell$ denotes $(\E|X|^\ell)^{1/\ell}$ and using Minkowski's inequality (that $\|\cdot\|_\ell$ is a norm for $\ell \ge 1$),
\allowdisplaybreaks
\begin{align*}
\|\|C'\|_F^2\|_{\ell/2} &= \left \| \sum_{i,j} \|(\Pi A'_i)^T(\Pi B'_j) - A_i'^T B'_j\|^2 \right \|_{\ell/2} \\
{}&\le \sum_{i,j} \|\|(\Pi A'_i)^T(\Pi B'_j) - A_i'^T B'_j\|^2 \|_{\ell/2}\\
{}&\le \sum_{i,j} \eps^2 s_i^2 t_j^2\cdot \delta^{2/\ell} \text{ (\Lemma{matmult})}\\
{}&= \eps^2 \left ( \sum_i s_i^2 \right ) \cdot \left ( \sum_j t_j^2 \right ) \delta^{2/\ell} \\
{}&\le \left ( \eps \sqrt{(\|A\|^2 + \|A\|_F^2 / k) (\|B\|^2 + \|B\|^2_F / k)} \delta^{1/\ell} \right )^2
\end{align*}

Now, $\E \|C'\|_F^\ell = \| \|C'\|_F^2 \|_{\ell/2}^{\ell/2}$, implying
\begin{align*}
\Pr \left ( \|C'\| > \eps \sqrt{(\|A\|^2 + \frac{\|A\|_F^2}k) (\|B\|^2 + \frac{\|B\|^2_F}k)} \right ) &\le \Pr \left ( \|C'\|_F > \eps \sqrt{(\|A\|^2 + \frac{\|A\|_F^2}k) (\|B\|^2 + \frac{\|B\|^2_F} k)} \right ) \\
{}&< \frac {\E\|C'\|_F^\ell}{\left ( \eps \sqrt{(\|A\|^2 + \frac{\|A\|_F^2}k) (\|B\|^2 + \frac{\|B\|^2_F}k)} \right)^\ell}\\
{}&\le \delta .
\end{align*}
\end{proof}

We now discuss the implications of applying \Theorem{main2} to specific OSE's.

\paragraph{Subgaussian maps:} In \Section{subgaussian-ose-moments} we show that if $\Pi$ has independent subgaussian entries and $m = \Omega((k+\log(1/\delta))/\eps^2)$ rows, then it satisfies the $(\eps,\delta,2k,\Theta(k + \log(1/\delta)))$ OSE moment property. Thus \Theorem{main2} applies to show that such $\Pi$ will satisfy $(k,\eps,\delta)$-AMM.

\paragraph{SRHT:} The SRHT is the matrix product $\Pi = SHD$ where $D\in\R^{n\times n}$ is $n\times n$ diagonal with independent $\pm 1$ entries on the diagonal, $H$ is a ``bounded orthonormal system'' (i.e.\ an orthogonal matrix in $\R^{n\times n}$ with $\max_{i,j} |H_{i,j}| = O(1/\sqrt{n})$), and the $m$ rows of $S$ are independent and each samples a uniformly random element of $[n]$. Bounded orthonormal systems include the discrete Fourier matrix and the Hadamard matrix; thus such $\Pi$ exist supporting matrix-vector multiplication in $O(n\log n)$ time. Thus when computing $\Pi A$ for some $n\times d$ matrix $A$, this takes time $O(nd\log n)$ (by applying $\Pi$ to $A$ column by column). In \Theorem{best-srht} we show that the SRHT with $m = \Omega((k + \log(1/(\eps\delta))\log(k/\delta))/\eps^2)$ satisfies the $(\eps, \delta, 2k, \log(k/\delta))$-OSE moment property, and thus provides $(k,\eps,\delta)$-AMM. Interestingly our analysis of the SRHT in \Section{srht} seems to be asymptotically tighter than any other analyses in previous work even for the basic subspace embedding property, and even slightly improves the by now standard analysis of the Fast JL transform given in \cite{AilonC09}.

\paragraph{Sparse subspace embeddings:} The sparse embedding distribution with parameters $m, s$ is as follows \cite{ClarksonW13,NelsonN13,KaneN14}. The matrix $\Pi$ has $m$ rows and $n$ columns. The columns are independent, and for each column exactly $s$ uniformly random entries are chosen without replacement and set to $\pm 1/\sqrt{s}$ independently; other entries in that column are set to zero. Alternatively, one could use the CountSketch \cite{CharikarCF04}: the $m$ rows are equipartitioned into $s$ sets of size $m/s$ each. The columns are independent, and in each column we pick exactly one row from each of the $s$ partitions and set the corresponding entry in that column to $\pm 1/\sqrt{s}$ uniformly; the rest of the entries in the column are set to $0$. Note $\Pi A$ can be multiplied in time $O(s\cdot \mathop{nnz}(A))$, and thus small $s$ is desirable.

It was shown in \cite{MengM13,NelsonN13}, slightly improving \cite{ClarksonW13}, that either of the above distributions satisfies the $(\eps,\delta,k, 2)$-OSE moment property for $m = \Omega(k^2/(\eps^2 \delta))$, $s = 1$, and hence $(k,\eps,\delta)$-AMM (though this particular conclusion follows easily from \cite[Theorem 6.2]{KaneN14}). It was also shown in \cite{Cohen16}, improving upon \cite{NelsonN13}, that they satisfy the $(\eps, \delta, k, \log(k/\delta))$-OSE moment property, and hence also $(k,\eps,\delta)$-AMM, for $m = \Omega(B k\log(k/\delta)/\eps^2), s = \Omega(\log_B(k/\delta)/\eps)$ for any $B > 2$. It is conjectured that for $B = O(1)$, $m = \Omega((k + \log(1/\delta))/\eps^2)$ should suffice \cite[Conjecture 14]{NelsonN13}.

\begin{remark}\RemarkName{cohen-osnap}
\textup{
The work \cite{Cohen16} does not explicitly discuss the OSE moment property for sparse subspace embeddings. Rather, \cite{Cohen16} bounds $\E e^{E 2\ell/\eps} = O(k)$ for $E = \|(\Pi U)^T(\Pi U) - I\|$ and $\ell = \log(k/\delta)$. Note though for $x\ge 0$ and integer $\ell \ge 1$, $x^\ell \le \ell!\cdot e^x\le \ell^\ell\cdot e^x$ by Taylor expansion of the exponential. Setting $x = 2\ell E/\eps$, \cite{Cohen16} thus implies $\E (2\ell E/\eps)^\ell \le \ell^\ell \cdot \E e^{E 2\ell/\eps} = O(k)$. Thus $\E E^\ell = O(k)\cdot (\eps/2)^\ell < \delta$ by choice of $\ell$, which is the $(\eps, \delta, k, \log(k/\delta))$-OSE moment property.
}
\end{remark}

\begin{remark}\RemarkName{efficiency}
\textup{
Currently there appears to be a tradeoff: one can either use $\Pi$ such that $\Pi A$ can be computed quickly, such as sparse subspace embeddings or the SRHT, but then the number of rows $m$ is at least $k\log k$. Alternatively one could achieve the optimal $m = O(k/\eps^2)$ using subgaussian $\Pi$, but then multiplying by $\Pi$ is slower: $O(mnd)$ time for $A\in\R^{n\times d}$. However, settling for a tradeoff is unnecessary. One can actually obtain the ``best of both worlds'' by composition, i.e.\ the multiplication $\Pi = \Pi_1\cdot \Pi_2$ of two matrices both supporting AMM. Thus $\Pi_2$ could be a fast matrix providing AMM to low (but suboptimal) dimension, and $\Pi_1$ a ``slow'' (e.g.\ subgaussian) matrix with the optimal $O(k/\eps^2)$ number of rows. In fact one can even set $\Pi = \Pi_1\Pi_2\Pi_3$ where $\Pi_3$ is the sparse subspace embedding with $O(k^2/\eps^2)$ rows and $s = 1$, $\Pi_2$ is the SRHT, and $\Pi_1$ is a subgaussian matrix. Then $\Pi A$ will have the desired $O(k/\eps^2)$ rows and can be computed in time $O(\mathop{nnz(A)}) + \tilde{O}(\eps^{-O(1)}(k^3 + k^2d))$; see \Section{composition} for justification.
}
\end{remark}

\subsection{Characterization for $(k,\eps,)$-AMM via deterministic events}\SectionName{conditioning}

Here we provide a different characterization for achieving $(k,\eps)$-AMM. Without loss of generality we assume $\max\{\|A\|^2, \|A\|_F^2 / k\} = \max\{\|B\|^2, \|B\|_F^2 / k\} = 1$ (so $\|A\|^2, \|B\|^2 \le 1$ and $\|A\|_F^2, \|B\|_F^2 \le k$).

Let $w,w'$ each be minimal such that $\|A_{\bar{w}}\|, \|B_{\bar{w'}}\| \le \eps / C'$ for some sufficiently large constant $C'$ (which will be set in the proof of \Theorem{main}). It was shown that $w,w' = O(k/\eps^2)$ in the proof of Theorem 3.2 (i.b) in \cite{MagenZ11}. Write the SVDs $A_w = U_{A_w} \Sigma_{A_w} V_{A_w}^T$, $B_{w'} = U_{B_{w'}} \Sigma_{B_{w'}} V_{B_{w'}}^T$.

For $0\le i\le \log_2(1/\eps^2)$ define $D_i'$ as set of all columns of $U_{A_w},U_{B_{w'}}$ whose corresponding squared singular values (from $\Sigma_{A_w},\Sigma_{B_{w'}}$) are at least $1/2^i$. Let $D_{A_w}$ be the set of $\min\{k,w\}$ largest singular vectors from $U_{A_w}$, and define $D_{B_{w'}}$ similarly. Define $D_i = D_i'\cup D_{A_w}\cup D_{B_{w'}}$. Let $s_i$ denote the dimension of $\mathrm{span}(D_i)$, and note the $s_i$ are non-decreasing.

Let $\tilde{s}_i$ be $s_i$ after rounding up to the nearest power of $2$. Group all $i$ with the same value of $\tilde{s}_i$ into groups $G_1,G_2,\ldots,G_{\log_2(1/\eps^2)}$. For example if for $i=0,1,2,3$ the $s_i$ are $3,4,15,16$ then the $\tilde{s}_i$ are $4,4,16,16$ and $G_1 = \{0,1\}$, $G_2 = \{2,3\}$. Let $v_j$ be the common value of $\tilde{s}_i$ for $i$ in $G_j$.

\begin{lemma}\LemmaName{si-sumbound}
$\sum_i s_i/2^i \le  8k$.
\end{lemma}
\begin{proof}
Define $s = |D_{A_w}\cup D_{B_{w'}}|\le 2k$ and let $s_i'$ denote the dimension of $\mathrm{span}(D_i')$. Then the above summation is at most $\sum_i (s/2^i + s_i'/2^i) \le 4k + \sum_i s_i'/2^i$. It thus suffices to bound the second summand by $4k$.

Note that we can find a basis for $D_i'$ among the columns of $U_{A_w}, U_{B_{w'}}$ with corresponding squared singular value at least $1/2^i$, so let $a_i + b_i = s_i'$, where $a_i$ is the number of columns of $U_{A_w}$ in the basis and $b_i$ the number of columns of $U_{B_{w'}}$ in the basis. Then by averaging, if the inequality of the lemma statement does not hold then either $\sum_i a_i/2^i > 2k$ or $\sum_i b_i/2^i > 2k$. Without loss of generality assume the former.

Consider an arbitrary column of $U_{A_w}$, and suppose it has squared singular value in the range $[1/2^i, 1/2^{i-1})$. Then it is in $\mathrm{span}(D_j')$ for all $j \ge i$. Its contribution to $\sum_i a_i/2^i$ is therefore $1/2^i + 1/2^{i+1} + \ldots$ which is at most $2/2^i = 1/2^{i-1}$. It follows that $\sum_i a_i/2^i \le 2k$, since the squared Frobenius norm of $A_w$ is at most $k$. This is a contradiction to $\sum_i a_i/2^i > 2k$.
\end{proof}

Now we prove the main theorem of this subsection.

\begin{theorem}\TheoremName{main}
Suppose that the following conditions hold:
\begin{enumerate}

\item[(1)] If $w+w' \le k$, then $\Pi$ is an $\eps/C$-subspace embedding for the subspace spanned by the columns of $A_w,B_{w'}$. Otherwise if $w+w' > k$, then for each $0\le i\le \log_2(1/\eps^2)$, $\Pi$ is an $\eps_i/C$-subspace embedding for $\mathrm{span}(D_{i'})$ with
$$
\eps_i = \min\left\{\frac 12, \eps \sqrt{\frac{v_j}{k}}\right\}
$$
where $i'$ is the largest $i$ with $s_i$ in $G_j$.
\item[(2)] $\|\Pi A_{\bar{w}}\|, \|\Pi B_{\bar{w'}}\| \le \eps/C$.
\end{enumerate}
Then \Equation{mmult-general} holds as long as $C$ is smaller than some fixed universal constant.
\end{theorem}
\begin{proof}
We would like to bound
\begin{align}
\nonumber \|(\Pi A)^T(\Pi B) - A^T B\| & \le \underbrace{\|(\Pi A_w)^T\Pi B_{w'} - A_w^T B_{w'}\|}_\alpha + \underbrace{\|(\Pi A_{\bar{w}})^T\Pi B_{w'}\|}_\beta + \underbrace{\|(\Pi A_w)^T\Pi B_{\bar{w'}}\|}_\gamma\\
&\hspace{.2in}{}+ \underbrace{\|(\Pi A_{\bar{w}})^T\Pi B_{\bar{w'}}\|}_\Delta + \underbrace{\|A_{\bar{w}}^T B_{w'}\|}_\zeta + \underbrace{\|A_w^T B_{\bar{w'}}\|}_\eta + \underbrace{\|A_{\bar{w}}^T B_{\bar{w'}}\|}_\Theta \EquationName{final-bound}
\end{align}

Using $\|XY\| \le \|X\|\cdot \|Y\|$ for any conforming matrices $X,Y$, we see $\Delta \le \eps^2/C^2$ by condition (2). Furthermore by the definition of $w,w'$ we know $\|A_{\bar{w}}\|, \|B_{\bar{w'}}\| \le \eps/C'$, and thus $\zeta + \eta +  \Theta \le 2\eps/C' + (\eps/C')^2$. Note condition (1) implies that $\Pi$ is a $(1/2)$-subspace embedding for the subspace spanned by columns of $A_w,B_{w'}$ (by taking $i$ maximal). Thus by both conditions we have $\beta, \gamma \le (\eps/C)(1+1/2)$.

It only remains to bound $\alpha$. If $w+w' \le k$, then we are done by condition (1) and \Lemma{simple-mmult}. Thus assume $w+w' > k$. Then we have
$$
\|(\Pi A_w)^T\Pi B_{w'} - A_w^T B_{w'}\| = \sup_{\|x\|=\|y\|=1} \left|\inprod{\Pi U_{A_w} \Sigma_{A_w} x, \Pi U_{B_{w'}} \Sigma_{B_{w'}} y} - \inprod{U_{A_w}\Sigma_{A_w} x, U_{B_{w'}} \Sigma_{B_{w'}} y}\right|
$$
Let $x,y$ be any unit norm vectors. Write $x = x^1 + x^2 + \ldots + x^b$ for $b = \log_2(1/\eps^2)$, where $x^i$ is the restriction of $x$ to coordinates for which the corresponding squared singular values in $\Sigma_{A_w}$ are in $(1/2^i, 1/2^{i-1}]$. Similarly define $y^1,\ldots,y^b$. Then $|\inprod{\Pi U_{A_w} \Sigma_{A_w} x, \Pi U_{B_{w'}} \Sigma_{B_{w'}} y} - \inprod{U_{A_w}\Sigma_{A_w} x, U_{B_{w'}} \Sigma_{B_{w'}} y}|$ equals
\allowdisplaybreaks
\begin{align}
\nonumber &\left|\sum_{i=1}^b\sum_{j=1}^b \inprod{\Pi U_{A_w}\Sigma_{A_w} x^i, \Pi U_{B_{w'}}\Sigma_{B_{w'}} y^j} - \inprod{U_{A_w}\Sigma_{A_w} x^i, U_{B_{w'}}\Sigma_{B_{w'}}y^j}\right|\\
\nonumber &\hspace{.4in}{}\le \sum_{i=1}^b \left|\inprod{\Pi U_{A_w}\Sigma_{A_w}x^i, \Pi U_{B_{w'}}\Sigma_{B_{w'}}\sum_{j\le i}y^j} - \inprod{U_{A_w}\Sigma_{A_w}x^i, \sum_{j\le i}U_{B_{w'}}\Sigma_{B_{w'}}y^j}\right|\\
&\hspace{.6in}{} + \sum_{j=1}^b \left|\inprod{\Pi U_{A_w}\Sigma_{A_w}\sum_{i\le j}x^i, \Pi U_{B_{w'}}\Sigma_{B_{w'}} y^j} - \inprod{\sum_{i\le j} x^i, y^j}\right| \EquationName{will-cs}
\end{align}
We bound the first sum, as bounding the second is similar. Note $U_{A_w}\Sigma_{A_w} x^i, U_{B_{w'}}\Sigma_{B_{w'}}\sum_{j\le i} y^j\in D_i$. Therefore by property (1) and \Lemma{simple-mmult},
\begin{align}
\nonumber \Bigg|\inprod{\Pi U_{A_w}\Sigma_{A_w}x^i, \Pi U_{B_{w'}}\Sigma_{B_{w'}}\sum_{j\le i}y^j}& - \inprod{U_{A_w}\Sigma_{A_w}x^i, U_{B_{w'}}\Sigma_{B_{w'}}\sum_{j\le i}y^j}\Bigg| \le \frac{\eps_i}{C 2^{(i-1)/2}} \cdot \|x^i\|\cdot \|y\|\\
{}&\le \frac{\eps}{C 2^{(i-1)/2}} \cdot \sqrt\frac{2s_i}{k} \cdot \|x^i\| \EquationName{vs}
\end{align}
where \Equation{vs} used that the corresponding $v$ value in property (1) is at most $2s_i$. Returning to \Equation{will-cs} and applying Cauchy-Schwarz and \Lemma{si-sumbound},
\begin{align*}
\sum_{i=1}^b \Bigg|\inprod{\Pi U_{A_w}\Sigma_{A_w}x^i, \Pi U_{B_{w'}}\Sigma_{B_{w'}}\sum_{j\le i}y^j}& - \inprod{U_{A_w}\Sigma_{A_w}x^i, \sum_{j\le i}U_{B_{w'}}\Sigma_{B_{w'}}y^j}\Bigg| \le \sum_{i=1}^b \frac{\eps}{C 2^{(i-1)/2}} \cdot \sqrt\frac{2s_i}{k} \cdot \|x^i\|\\
{}&\le \frac{2\eps}{C\sqrt{k}}\cdot \left(\sum_{i=1}^b \frac{s_i}{2^i}\right)^{1/2} \cdot \left(\sum_{i=1}^b \|x^i\|^2\right)^{1/2} \\
{}&\le \frac{2\sqrt{8}\eps}C
\end{align*}

We thus finally have that \Equation{final-bound} is at most $(2\sqrt{8} + 3)\eps/C + + (\eps/C)^2 + 2\eps/C' + (\eps/C')^2$, which is at most $\eps$ for $C, C'$ sufficiently large constants.
\end{proof}

Now we discuss some implications of \Theorem{main} for specific $\Pi$.

\paragraph{Example 1:} Let $\Pi$ have $O(k/\eps^2)$ rows forming an orthonormal basis for the span of the columns of $A_w,B_{w'}$. Property (1) is satisfied for every $i$ in fact with $\eps_i = 0$. Property (2) is also satisfied since $\|\Pi A_{\bar{w}}\| \le \|\Pi\| \cdot \|A_{\bar{w}}\| \le \eps$, and similarly for bounding $\|\Pi B_{\bar{w'}}\|$.

\paragraph{Example 2:} Let $\Pi$ be a random $m\times n$ matrix with independent entries that are subgaussian with variance $1/m$. For example, the entries of $\Pi$ may be $\mathcal{N}(0,1/m)$, or uniform in $\{-1/\sqrt{m}, 1/\sqrt{m}\}$. Let $m$ be $\Theta((k+\log(1/\delta))/\eps^2)$. As mentioned in \Section{subgaussian-ose-moments}, such $\Pi$ is an $\eps$-subspace embedding for a $k$-dimensional subspace with failure probability $\delta$.  For property (1) of \Theorem{main}, if $w+w'\le k$ then we would like $\Pi$ to be an $\eps$-subspace embedding for a subspace of dimension at most $k$, which holds with failure probability $\delta$. If $w+w' > k$ then we would like $\Pi$ to be an $\eps_i$-subspace embedding for $\mathrm{span}(D_{i'})$ for all $1\le i\le \log_2(1/\eps^2)$ simultaneously. Note $\max_j v_j \le 2(w+w') = O(k/\eps^2)$, and thus $\max_j v_j \le m$. Thus for a subspace under consideration $\mathrm{D_{i'}}$ for $i' \in G_j$, we have failure probability $\delta^{v_j/k}$ for our choice of $m$. By construction every $v_j$ is at least $k$, and the $v_j$ increase at least geometrically. Thus our total failure probability is, by a union bound, $\sum_j \delta^{v_j/k} \le \sum_j \delta^{2^{j-1}} = O(\delta)$. Property (2) of \Theorem{main} is satisfied with failure probability $\delta$ by \cite[Theorem 3.2]{RudelsonV13}.

\section{Applications}\SectionName{applications}
Spectral norm approximate matrix multiplication with dimension bounds depending on stable rank has immediate applications for the analysis of generalized regression and low-rank approximation problems. We also point out to the reader recent applications of this result to kernelized ridge regression \cite{YangPW15} and $k$-means clustering \cite{CohenEMMP14}.

\subsection{Generalized regression}\SectionName{gen-reg}

Here we consider generalized regression: attempting to approximate a matrix $B$ as $AX$, with $A$ of rank at most $k$.  Let $P_A$ be the orthogonal projection operator to the column space of $A$, with $P_{\bar A} = I-P$; then the natural best approximation
will satisfy
$$
AX = P_A B.
$$
This minimizes both the Frobenius and spectral norms of $AX - B$.  A standard approximation algorithm for this is to replace $A$ and $B$ with sketches $\Pi A$ and $\Pi B$, then
solve the reduced problem exactly (see e.g. \cite{ClarksonW09}, Theorem 3.1).  This will produce 
\begin{align*}
\tilde X &= ((\Pi A)^T \Pi A)^{-1} (\Pi A)^T \Pi B \\
A \tilde X &= A ((\Pi A)^T \Pi A)^{-1} (\Pi A)^T \Pi B \\
{}&= U_A ((\Pi U_A)^T\Pi U_A)^{-1} (\Pi U_A)^T\Pi B.
\end{align*}

Below we give a lemma on the guarantees of the sketched solution in terms of properties of $\Pi$.

\begin{theorem}\TheoremName{gen-regression}
If $\Pi$
\begin{enumerate}
\item
satisfies the $(k, \sqrt{\eps / 8})$-approximate spectral norm matrix multiplication property for $U_A, P_{\bar A} B$
\item
is a $(1/2)$-subspace embedding for the column space of $A$ (which is implied by $\Pi$ satisfying the spectral norm approximate matrix multiplication property for $U_A$ with itself)
\end{enumerate}
then
\begin{equation}
\| A \tilde X - B \|^2 \leq (1 + \eps) \| P_A B - B \|^2 + (\eps/k)\cdot \| P_A B - B \|_F^2.\EquationName{regress}
\end{equation}
\end{theorem}
\begin{proof}
We may write:
\begin{align*}
\| A \tilde X - B \|_2^2 &= \| U_A ((\Pi U_A)^T\Pi U_A)^{-1} (\Pi U_A)^T \Pi B - B \|^2 \\
{}&= \| U_A ((\Pi U_A)^T\Pi U_A)^{-1} (\Pi U_A)^T \Pi ( P_A B + P_{\bar A} B ) - P_A B - P_{\bar A} B \|^2 \\
{}&= \| P_A B + U_A ((\Pi U_A)^T \Pi U_A)^{-1} (\Pi U_A)^T\Pi P_{\bar A} B - P_A B - P_{\bar A} B \|^2 \\
{}&= \| U_A ((\Pi U_A)^T \Pi U_A)^{-1} (\Pi U_A)^T\Pi P_{\bar A} B - P_{\bar A} B \|^2.
\end{align*}
So far, we have shown that the error depends only on $P_{\bar A} B$ and not $P_A B$ (with the third line following from the fact that the sketched regression is exact on $P_A B$).  Now, in the last line, we can see that the two terms lie in orthogonal column spaces (the first in the span of $A$, the second orthogonal to it).  For matrices $X$ and $Y$ with orthogonal column spans, $\| X+Y \|^2 \leq \| X \|^2 + \| Y \|^2$, so this is at most
$$
\| U_A ((\Pi U_A)^T \Pi U_A)^{-1} (\Pi U_A)^T\Pi P_{\bar A} B \|^2 + \| P_{\bar A} B \|^2.
$$
Spectral submultiplicativity then implies the first term is at most
$$
( \| U_A \| \cdot \| ((\Pi U_A)^T \Pi U_A)^{-1} \| \cdot \| (\Pi U_A)^T \Pi P_{\bar A} B \| )^2.
$$
$\| U_A \|$ is 1, since $U_A$ is orthonormal.  $((\Pi U_A)^T\Pi U_A)^{-1}$ is at most 2, since $\Pi$ is a subspace embedding for $U_A$.  Finally, $\| (\Pi U_A)^T \Pi P_{\bar A} B \|$ is at most
$$
\sqrt{\eps / 8} \cdot \sqrt{(\|U_A\|^2 + \|U_A\|_F^2 / k) (\|P_{\bar A} B\|^2 + \|P_{\bar A} B\|^2_F / k)} = \sqrt{(\eps / 8) \cdot 2 \cdot ( \| P_A B - B \|^2 + \| P_A B - B \|^2 / k )}.
$$
Multiplying these together, squaring, and adding the remaining $\| P_{\bar A} B \|^2$ term gives a bound of
$$
(1 + \eps) \| P_A B - B \|^2 + (\eps / k)\cdot \| P_A B - B \|_F^2
$$
as desired.
\end{proof}

\subsection{Low-rank approximation}

Now we apply the generalized regression result from \Section{gen-reg} to obtain a result on low-rank approximation: approximating a matrix $A$ in the form $\tilde U_k \tilde \Sigma_k \tilde V_k^T$, where $\tilde U_k$ has only $k$ columns and both $\tilde U_k$ and $\tilde V_k$ have orthonormal columns.  Here, we consider a previous approach (see e.g.\ \cite{Sarlos06}):
\begin{enumerate}
\item
Let $S = \Pi A$.  
\item
Let $P_S$ be the orthogonal projection operator to the row space of $S$.  Let $\tilde A = A P_S$.
\item
Compute a singular value decomposition of $\tilde A$, and keep only the top $k$ singular vectors.  Return the resulting low rank approximation $\tilde{A}_k$ of $\tilde A$.
\end{enumerate}
It turns out computing $\tilde{A}_k$ can be done much more quickly than computing $A_k$; see details in \cite[Lemma 4.3]{ClarksonW09}.

Let $A_k$ be the exact $k$-truncated SVD approximation of $A$ (and thus the best rank-$k$ approximation, in the spectral and Frobenius norms), and let  $U_k$ be the top $k$ column singular vectors, and $A_{\bar k} = A - A_k$ be the tail. 

\begin{theorem}
If $\Pi$
\begin{enumerate}
\item
satisfies the $(k, \sqrt{\eps / 8})$-approximate spectral norm matrix multiplication property for $U_k, A_{\bar k}$
\item
is a $(1/2)$-subspace embedding for the column space of $U_k$
\end{enumerate}
then
\begin{equation}
\| A - \tilde{A}_k \|^2 \leq (1 + \eps) \| A - A_k \|^2 + (\eps / k) \| A - A_k \|_F^2
\end{equation}
\end{theorem}
\begin{proof}
Note that this procedure chooses the best possible (in the spectral norm) rank-$k$ approximation to $A$ subject to the constraint of lying in the row space of $S$.  Thus, the spectral norm error can be no worse than the error of a specific such matrix we exhibit.

We simply choose the matrix obtained by running our generalized regression algorithm from $A$ onto $U_k$, with $\Pi$:
$$
U_k ((\Pi U_k)^T \Pi U_k)^{-1} (\Pi U_k)^T \Pi A
$$
This is rank-$k$ by construction, since it is multiplied by $U_k$, and it lies in the row space of $S = \Pi A$ since that is the rightmost factor.  On the other hand, it is an application of the regression algorithm to $A$ where the optimum output is $A_k$ (since that is the projection of $A$ onto the space of $U_k$).  Plugging this into \Equation{regress} gives the desired result.
\end{proof}

\subsection{Kernelized ridge regression}
In nonparametric regression one is given data $y_i = f^*(x_i) + w_i$ for $i=1,\ldots,n$, and the goal is to recover a good estimate for the function $f^*$. Here the $y_i$ are scalars, the $x_i$ are vectors, and the $w_i$ are independent noise, often assumed to be distributed as mean-zero gaussian with some variance $\sigma^2$. Unlike linear regression where $f^*(x_i)$ is assumed to take the form $\inprod{\beta, x}$ for some vector $\beta$, in nonparametric regression we allow $f^*$ to be an arbitrary function from some function space. Naturally the goal then is to recover some $\tilde{f}$ from the data so that, as $n$ grows, the probability that $\tilde{f}$ is ``close'' to $f^*$ increases at some good rate.

The recent work \cite{YangPW15} considers the well studied problem of obtaining $\tilde{f}$ so that $\|\tilde{f} - f^*\|_n^2$ is small with high probability over the noise $w$, where one uses the definition
$$
\|f - g\|_n^2 = \frac 1n\sum_{i=1}^n (f(x_i) - g(x_i))^2 .
$$
The work \cite{YangPW15} considers the case where $f^*$ comes from a Hilbert space $\mathcal{H}$ of functions $f$ such that $f$ is guaranteed to be square integrable, and the map $x\mapsto f(x)$ is a bounded linear functional. The function $\tilde{f}$ is then defined to be the optimal solution to the {\em Kernel Ridge Regression (KRR)} problem of computing
\begin{equation}
f^{LS} = \mathop{argmin}_{f\in\mathcal{H}}\left\{ \frac 1{2n}\sum_{i=1}^n (y_i - f(x_i))^2 + \lambda_n \cdot \|f\|_{\mathcal{H}}^2\right\} \EquationName{krr}
\end{equation}
for some parameter $\lambda_n$. It is known that any $\mathcal{H}$ as above can be written as the closure of the set of all functions
\begin{equation}
g(\cdot) = \sum_{i=1}^N \alpha_i k(\cdot, z_i) ,\EquationName{kernelized}
\end{equation}
over all $\alpha\in\R^N$ and vectors $z_1,\ldots,z_N$ for some positive semidefinite {\em kernel function} $k$. Furthermore, the optimal solution to \Equation{krr} can be expressed as $f^{LS} = \sum_{i=1}^n \alpha^{LS}_i\cdot k(\cdot, x_i)$ for some choice of weight vector $\alpha^{LS}$, and it is known that $\|f^{LS} - f^*\|_n$ will be small with high probability, over the randomness in $w$, if $\lambda_n$ is chosen appropriately (see \cite{YangPW15} for background references and precise statements).

After rewriting \Equation{krr} using \Equation{kernelized} and defining a matrix $K$ with $K_{i,j} = k(x_i, x_j)$, one arrives at a reformulation for KRR of computing 
$$
\alpha^{LS} = \mathop{argmin}_{\alpha\in\R^n}\left\{ \frac 1{2n} \alpha^T K^2 \alpha - \frac 1n\alpha^T K y + \lambda_n \alpha^T K \alpha\right\} = \left(\frac 1n K^2 + 2\lambda_n K\right)^{-1}\cdot \frac 1n Ky ,
$$ 
which can be computed in $O(n^3)$ time. The work \cite{YangPW15} then focuses on speeding this up, by instead computing a solution to the lower-dimensional problem
$$
\tilde{\alpha}^{LS} = \mathop{argmin}_{\alpha\in\R^m}\left\{ \frac 1{2n} \alpha^T \Pi K^2\Pi^T \alpha - \frac 1n\alpha^T \Pi K y + \lambda_n \alpha^T \Pi K \Pi^T \alpha\right\} = \left(\frac 1n \Pi K^2\Pi^T + 2\lambda_n \Pi K\Pi^T\right)^{-1}\cdot \frac 1n \Pi Ky 
$$
and then returning as $\tilde{f}$ the function specified by the weight vector $\tilde{\alpha} = \Pi^T \tilde{\alpha}^{LS}$. Note that once various matrix products are formed (where the running time complexity depends on the $\Pi$ being used), one only needs to invert an $m\times m$ matrix thus taking $O(m^3)$ time. They then prove that $\|\tilde{f} - f^*\|_n$ is small with high probability as long as $\Pi$ satisfies two deterministic conditions (see the proof of Lemma 2 \cite[Section 4.1.2]{YangPW15}, specifically equation (26) in that work):
\begin{itemize}
\item $\Pi$ is a $(1/2)$-subspace embedding for a particular low-dimensional subspace
\item $\|\Pi B\| = O(\|B\|)$ for a particular matrix $B$ of low stable rank ($B$ is $U D_2$ in \cite{YangPW15}). Note 
$$
\|\Pi B\| = \|(\Pi B)^T\Pi B\|^{1/2} \le \left(\|(\Pi B)^T\Pi B - B^T B\| + \|B^T B\|\right)^{1/2} \le \|(\Pi B)^T\Pi B - B^T B\|^{1/2} + \|B\| ,
$$
and thus it suffices for $\Pi$ to provide the approximate matrix multiplication property for the product $B^T B$, where $B$ has low stable rank.
\end{itemize}
The first bullet simply requires a subspace embedding in the standard sense, and for the second bullet \cite{YangPW15} avoided AMM by obtaining a bound on $\|\Pi B\|$ directly by their own analyses for gaussian $\Pi$ and the SRHT (in the gaussian case, it also follows from \cite[Theorem 3.2]{RudelsonV13}). Our result thus provides a unifying analysis which works for a larger and general class of $\Pi$, including for example sparse subspace embeddings.

\subsection{$k$-means clustering}

In the works \cite{BoutsidisZMD15,CohenEMMP14}, the authors considered dimensionality reduction methods for $k$-means clustering. Recall in $k$-means clustering one is given $n$ points $x_1,\ldots,x_n\in\R^d$, as well as an integer $k\ge 1$, and the goal is to find $k$ points $y_1,\ldots,y_k\in\R^d$ minimizing
$$
\sum_{i=1}^n \min_{j=1}^k \|x_i - y_j\|_2^2 .
$$
That is, the $n$ points can be partitioned arbitrarily into $k$ clusters, then a ``cluster center'' should be assigned to each cluster so as to minimize sums of squared Euclidean distances of each of the $n$ points to their cluster centers. It is a standard fact that once a partition $\mathcal{P} = \{P_1,\ldots,P_k\}$ of the $n$ points into clusters is fixed, the optimal cluster centers to choose are the centroids of the points in each of the $k$ partitions, i.e.\ $y_j = (1/|P_j|)\cdot \sum_{i \in P_j} x_i$.

One key observation common to both of the works \cite{BoutsidisZMD15,CohenEMMP14} is that $k$-means clustering is closely related to the problem of low-rank approximation. More specifically, given a partition $\mathcal{P} = \{P_1, \ldots, P_k\}$, define the $n\times k$ matrix $X_{\mathcal{P}}$ by
$$
(X_{\mathcal{P}})_{i, j} =
\begin{cases}
\frac 1{\sqrt{|P_j|}}, &\ \text{if } i \in P_j\\
0, &\ \text{otherwise}
\end{cases}
$$
Let $A\in\R^{n\times d}$ have rows $x_1,\ldots,x_n$. Then the $k$-means problem can be rewritten as computing
$$
\mathrm{argmin}_{\mathcal{P}} \|A - X_{\mathcal{P}} X_{\mathcal{P}}^T A\|_F^2
$$
where $\mathcal{P}$ ranges over all partitions of $\{1,\ldots,n\}$ into $k$ sets. It is easy to verify that the non-zero columns of $X_{\mathcal{P}}$ are orthonormal, so $X_{\mathcal{P}}X_{\mathcal{P}}^T$ is the orthogonal projection onto the column space of $X_{\mathcal{P}}$. Thus if one defines $\mathcal{S}$ as the set of all rank at most $k$ orthogonal projections obtained as $X_{\mathcal{P}}X_{\mathcal{P}}^T$ for some $k$-partition $\mathcal{P}$, then the above can be rewritten as the {\em constrained rank-$k$ projection problem} of computing
\begin{equation}
\mathrm{argmin}_{P\in\mathcal{S}} \|(I-P)A\|_F^2 .
\end{equation}
One can verify this by hand, since the rows of $A$ are the points $x_i$, and the $i$th row of $PA$ for $P = X_{\mathcal{P}}X_{\mathcal{P}}^T$ is the centroid of the points in $i$'s partition in $\mathcal{P}$.

The work \cite{CohenEMMP14} showed that if $\mathcal{S}$ is any subset of projections of rank at most $k$ (henceforth {\em rank-$k$ projections}) and $\Pi\in\R^{m\times d}$ satisfies certain technical conditions to be divulged soon, then if $\tilde{P}\in\mathcal{S}$ satisfies
\begin{equation}
\|(I-\tilde{P})A\Pi^T\|_F^2 \le \gamma \cdot \mathrm{min}_{P\in\mathcal{S}} \|(I-P)A\Pi^T\|_F^2 , \EquationName{pre-condition}
\end{equation}
then
\begin{equation}
\|(I-\tilde{P})A\|_F^2 \le \frac{(1+\eps)}{(1-\eps)}\cdot  \gamma \cdot \mathrm{min}_{P\in\mathcal{S}} \|(I-P)A\|_F^2 . \EquationName{post-condition}
\end{equation}

One set of sufficient conditions for $\Pi$ is as follows (see \cite[Lemma 10]{CohenEMMP14}). Let $A_k$ denote the best rank-$k$ approximation to $A$ and let $A_{\bar{k}} = A - A_k$. Define $Z \in \R^{d\times r}$ for $r = 2k$ by $Z = V_r$, i.e.\ the top $r$ right singular vectors of $A$ are the columns of $Z$. Define $B_1 = Z^T$ and $B_2 = \frac{\sqrt{k}}{\|A_{\bar{k}}\|_F}\cdot (A - AZZ^T)$. Define $B\in\R^{(n+r)\times d}$ as having $B_1$ as its first $r$ rows and $B_2$ as its lower $n$ rows. Then \cite[Lemma 10]{CohenEMMP14} states that \Equation{pre-condition} implies \Equation{post-condition} as long as
\begin{align}
\|(\Pi B^T)^T (\Pi B^T) - BB^T\| &< \eps, \EquationName{condition1-sr}\\
\text{and }\left| \|\Pi B_2\|_F^2 - \|B_2\|_F^2\right| & \le \eps k \EquationName{easier-condition}
\end{align}

One can easily check $\|B\|^2 = 1$ and $\|B\|_F^2 \le 3k$, so the stable rank $\sr(B)$ is at most $3k$. Thus \Equation{condition1-sr} is implied by the $(3k, \eps/2)$-AMM property for $B^T, B^T$, and our results apply to show that $\Pi$ can be taken to have $m = O((k + \log(1/\delta))/\eps^2)$ rows to have success probability $1-\delta$ for \Equation{condition1-sr}. Obtaining \Equation{easier-condition} is much simpler and can be derived from the JL moment property (see the proof of \cite[Theorem 6.2]{KaneN14}).

Without our results on stable-rank AMM provided in this current work, \cite{CohenEMMP14} gave a different analysis, avoiding \cite[Lemma 10]{CohenEMMP14}, which required $\Pi$ to have $m = \Theta(k\cdot \log(1/\delta)/\eps^2)$ rows (note the product between $k$ and $\log(1/\delta)$ instead of the sum). 

\section{Stable rank and row selection}\SectionName{bss}
As well as random projections, approximate matrix multiplication (and subspace embeddings) by row selection are also common in algorithms.  This corresponds to setting $\Pi$ to a diagonal matrix $S$ with relatively few nonzero entries.  Unlike random projections, there are no \emph{oblivious} distributions of such matrices $S$ with universal guarantees.  Instead, $S$ must be determined (either randomly or deterministically) from the matrices being embedded.

There are two particularly algorithmically useful methods for obtaining such $S$.  The first is importance sampling: independent random sampling of the rows, but with nonuniform sampling probabilities.  This is analyzed using matrix Chernoff bounds \cite{AhlswedeW02}, and for the case of $k$-dimensional subspace embedding or approximate matrix multiplication of rank-k matrices, it can produce $O(k (\log k) / \eps^2)$ samples \cite{SpielmanS11}. The second method is the deterministic selection method given in \cite{BatsonSS12}, often called ``BSS'', choosing only $O(k / \eps^2)$ rows.  This still runs in polynomial time, but originally required many relatively expensive linear algebra steps and thus was slower in general; see \cite{LeeS15} for runtime improvements.

The matrix Chernoff methods can be extended to the stable-rank case, making even the log factor depend only on the stable rank, using ``intrinsic dimension'' variants of the bounds as presented in Chapter 7 of \cite{Tropp15}.  Specifically, Theorem 6.3.1 of that work can be applied with each $n$ summands each equal to $\frac{1}{n} \left ( \frac{1}{p_i} a_i^T b_i - A^T B \right )$, where $a_i$ is the $i$th row of $A$, and $i$ is random with the probability of choosing a particular row $i$ equal to
\begin{equation*}
p_i = \frac{\| a_i \|^2 + \| b_i \|^2}{\sum_j \| a_j \|^2 + \| b_j \|^2}
\end{equation*}

We here give an extension of BSS that covers low stable rank matrices as well.
\begin{theorem}\TheoremName{bss-thm}
Given an $n$ by $d$ matrix $A$ such that $\| A \|^2 \le 1$ and $\| A \|_F^2 \le k$, and an $\eps \in (0, 1)$, there exists a diagonal matrix $S$ with $O(k / \eps^2)$ nonzero entries such that
$$
\| (SA)^T (SA) - A^T A \| \le \eps
$$
Such an $S$ can be computed by a polynomial-time algorithm.
\end{theorem}

When $A^TA$ is the identity, this is just the original BSS result.  It is also stronger than Theorem 3.3 of \cite{KollaMST10}, implying it when $A$ is the combination of the rows $\sqrt{N/T}\cdot v_i$ from that theorem statement with an extra column containing the costs, and a constant $\epsilon$.  The techniques in that paper, on the other hand, can prove a result comparable to \Theorem{bss-thm}, but with the row count scaling as $k/\eps^3$ rather than $k/\eps^2$.

\medskip

\begin{proof}
The proof closely follows the original proof of BSS.  However, for simplicity, and because the tight constants are not needed for most applications, we do not include \cite[Claim 3.6]{BatsonSS12} and careful parameter-setting.

At each step, the algorithm will maintain a partial approximation $Z = (SA)^T (SA)$ (the matrix ``$A$'' in \cite{BatsonSS12}), with $S$ beginning as 0.  Additionally, we keep track of upper and lower ``walls'' $X_u$ and $X_l$; in the original BSS these are just multiples of the identity.  The final $S$ will be returned by the algorithm (rescaled by a constant so that the average of the upper and lower walls is $A^T A$).

We will maintain the invariants
\begin{align}
\Tr(A (X_u - Z)^{-1} A^T) &\le 1 \\
\Tr(A (Z - X_l)^{-1} A^T) &\le 1.
\end{align}
These are the so-called upper and lower potentials from BSS. We also require $X_u \prec Z \prec X_l$; recall $M\prec M'$ means that $M'-M$ is positive definite.  Note that unlike \cite{BatsonSS12}, here we do not apply a change of variables making $A^T A$ the identity (to avoid confusion, since that would change the Frobenius norm).  This is the reason for the slightly more complicated form of the potentials.

In the original BSS, $X_u$ and $X_l$ were always scalar multiples of the identity (here, without the change of variables, that would correspond to always being multiples of $A^T A$).  \cite{BatsonSS12} thus simply represented them with scalars.  Like BSS, we will increase $X_u$ and $X_l$ by multiples of $A^T A$--however, the key difference from BSS is that they are \emph{initialized} to multiples of the identity, rather than $A^T A$.  In particular, we may initialize $X_u$ to $k I$ and $X_l$ to $-k I$.  This is still good enough to get the spectral norm bounds we require here (as opposed to the stronger multiplicative approximation guaranteed by BSS).

We will have two scalar values, $\delta_u$ and $\delta_l$, depending only on $\eps$; they will be set later.  One step consists of
\begin{enumerate}
\item
Choose a row $a_i$ from $A$ and a positive scalar $t$, and add $t a_i a_i^T$ to $Z$ (via increasing the $i$ component of $S$).
\item
Add $\delta_u A^T A$ to $X_u$ and $\delta_l A^T A$ to $X_l$.
\end{enumerate}
We will show that with suitable values of $\delta_u$ and $\delta_l$, for any $Z$ obeying the invariants there always exists a choice of $i$ and $t$ such that the invariants will still be true after the step is complete.  This corresponds to Lemmas 3.3 through 3.5 of BSS.

For convenience, we define, at a given step, the matrix functions of $y$
\begin{align*}
M_u(y) &= ((X_u + y A^T A) - Z)^{-1} \\
M_l(y) &= (Z - (X_l + y A^T A))^{-1}.
\end{align*}

The upper barrier value, after making a step of $t a_i a_i^T$ and increasing $X_u$, is
\begin{equation*}
\Tr(A ((X_u + \delta_u A^T A) - (Z + t a_i a_i^T))^{-1} A^T).
\end{equation*}
Applying the Sherman-Morrison formula, and cyclicity of trace, to the rank-1 update $t a_i a_i^T$, this can be rewritten as
\begin{equation*}
\Tr(A M_u(\delta_u) A^T) + \frac{t a_i^T M_u(\delta_u) A^T A M_u(\delta_u) a_i}{1 - t a_i^T M_u(\delta_u) a_i}.
\end{equation*}
Since the function $f(y) = \Tr(A M_u(y) A^T)$ is a convex function of $y$ with derivative 
$$f'(y) = -\Tr(A M_u(y) A^T A M_u(y) A^T) ,$$
we have $f(\delta_u) - f(0) \le -\delta_u \Tr(A M_u(\delta_u) A^T A M_u(\delta_u) A^T)$. Then the difference between the barrier before and after the step is at most
\begin{align*}
&\frac{t a_i^T M_u(\delta_u) A^T A M_u(\delta_u) a_i}{1 - t a_i^T M_u(\delta_u) a_i} - \delta_u \Tr(A M_u(\delta_u) A^T A M_u(\delta_u) A^T).
\end{align*}
Constraining this to be no greater than zero, rewriting in terms of $\frac{1}{t}$ and pulling it out gives
\begin{equation*}
\frac{1}{t} \ge \frac{a_i^T M_u(\delta_u) A^T A M_u(\delta_u) a_i}{\delta_u \Tr(A M_u(\delta_u) A^T A M_u(\delta_u) A^T)} + a_i^T M_u(\delta_u) a_i.
\end{equation*}
Furthermore, as long as $\frac{1}{t}$ is at least this, $Z$ will remain below $X_u$, since the barrier must approach infinity as $t$ approaches the smallest value passing $X_u$.

For the lower barrier value after the step, we get
\begin{equation*}
\Tr(A ((Z + t a_i a_i^T) - (X_l + \delta_l A^T A))^{-1} A^T).
\end{equation*}
Again, applying Sherman-Morrison rewrites it as
\begin{equation*}
\Tr(A M_l(\delta_l) A^T) - \frac{t a_i^T M_l(\delta_l) A^T A M_l(\delta_l) a_i}{1 + t a_i^T M_l(\delta_l) a_i}.
\end{equation*}
Again, due to convexity the increase in the barrier from raising $X_l$ is at most $\delta_l$ times the local derivative.  The difference in the barrier after the step is then at most
\begin{align*}
&-\frac{t a_i^T M_l(\delta_l) A^T A M_l(\delta_l) a_i}{1 + t a_i^T M_l(\delta_l) a_i} + \delta_l \Tr(A M_l(\delta_l) A^T A M_l(\delta_l) A^T).
\end{align*}
This is not greater than zero as long as
\begin{equation*}
\frac{1}{t} \le \frac{a_i^T M_l(\delta_l) A^T A M_l(\delta_l) a_i}{\delta_l \Tr(A M_l(\delta_l) A^T A M_l(\delta_l) A^T)} - a_i^T M_l(\delta_l) a_i.
\end{equation*}

There is some value of $t$ that works for $a_i$ as long as the lower bound for $\frac{1}{t}$ is no larger than the upper bound.  To show that there is at least one choice of $i$ for which this holds, we look at the sum of all the lower bounds and compare to the sum of all the upper bounds.  Summing the former over all $i$ gets
\begin{equation*}
\frac{\Tr(A M_u(\delta_u) A^T A M_u(\delta_u) A^T)}{\delta_u \Tr(A M_u(\delta_u) A^T A M_u(\delta_u) A^T)} + \Tr(A M_u(\delta_u) A^T)
\end{equation*}
and the latter gets
\begin{equation*}
\frac{\Tr(A M_l(\delta_l) A^T A M_l(\delta_l) A^T)}{\delta_l \Tr(A M_l(\delta_l) A^T A M_l(\delta_l) A^T)} - \Tr(A M_l(\delta_l) A^T).
\end{equation*}

Finally, note that 
\begin{equation*}
\Tr(A M_u(\delta_u) A^T) = \Tr(A ((X_u + \delta_u A^T A) - Z)^{-1} A^T) \le \Tr(A (X_u - Z)^{-1} A^T) \le 1
\end{equation*}
and the lower barrier implies $Z - X_l \succ A^T A$, implying that as long as $\delta_l \leq \frac{1}{2}$,
\begin{equation*}
\Tr(A M_l(\delta_l) A^T) = \Tr(A (Z - (X_l + \delta_l A^T A))^{-1} A^T) \le 2 \Tr(A (Z - X_l)^{-1} A^T) \le 2.
\end{equation*}

Thus, we can always make a step as long as $\delta_u$ and $\delta_l$ are set so that
\begin{equation*}
\frac{1}{\delta_u} + 1 \leq \frac{1}{\delta_l} - 2
\end{equation*}
and $\delta_l \leq \frac{1}{2}$.
This is satisfied by
\begin{align*}
\delta_u &= \eps + 2 \eps^2 \\
\delta_l &= \eps - 2 \eps^2.
\end{align*}

Before the first step, $X_u$ and $X_l$ can be initialized as $kI$ and $-kI$, respectively.  If the algorithm is then run for $\frac{k}{\eps^2}$ steps, we have:
\begin{align*}
X_u &= \frac{k}{\eps} A^T A + 2k A^T A + k I \\
&\preceq \frac{k}{\eps} A^T A + 3k I \\
X_l &= \frac{k}{\eps} A^T A - 2k A^T A - k I \\
&\succeq \frac{k}{\eps} A^T A - 3k I.
\end{align*}
$\frac{\eps}{k} X_u$ and $\frac{\eps}{k} X_l$ both end up within $3 \eps I$ of $A^T A$, so $\frac{\eps}{k} Z$ (from $\sqrt{\frac{\eps}{k}} S$) satisfies the requirements of the output for $3 \eps$ (one can simply apply this argument for $\eps / 3$).  Furthermore, all the computations required to verify the preservation of invariants and compute explicit $t$s can be performed in polynomial time.
\end{proof}

This obtains more general AMM as a corollary:
\begin{corollary}\CorollaryName{bss-amm}
Given two matrices $A$ and $B$, each with $n$ rows, and an $\eps \in (0, 1)$, there exists a diagonal matrix $S$ with $O(k / \eps^2)$ nonzero entries satisfying the $(k, \eps)$-AMM property for $A$, $B$.
Such an $S$ can be computed by a polynomial-time algorithm.
\end{corollary}
\begin{proof}
Apply \Theorem{bss-thm} to a matrix $X$ consisting of the columns of $\frac{A}{\sqrt{2} \max(\| A \|_2, \| A \|_F / \sqrt{k})}$ appended to the columns of $\frac{B}{\sqrt{2} \max(\| B \|_2, \| B \|_F / \sqrt{k})}$, and use the resulting $S$.

Note that $X$ satisfies the conditions of that theorem, since concatenating the sets of columns at most adds the squares of their spectral and Frobenius norms.  $(SA)^T (SB) - A^T B$ is a submatrix of $2 \max(\| A \|_2, \| A \|_F / \sqrt{k}) \max(\| B \|_2, \| B \|_F / \sqrt{k}) ((SX)^T (SX) - X^T X)$, so its spectral norm is upper bounded by the spectral norm of that matrix, which in turn is bounded by the guarantee of \Theorem{bss-thm}.
\end{proof}

\section*{Acknowledgments}
We thank Jaros{\l}aw B{\l}asiok for pointing out the connection between low stable rank approximate matrix multiplication and the analyses in \cite{YangPW15}.

\bibliographystyle{alpha}
\bibliography{../../biblio}

\appendix

\section*{Appendix}

\section{OSE moment property}\SectionName{ose-moment-property}
In the following two subsections we show the OSE moment property for both subgaussian matrices and the SRHT.

\subsection{Subgaussian matrices}\SectionName{subgaussian-ose-moments}
In this section, we show the OSE moment property for distributions satisfying a JL condition, namely the JL moment property. This includes matrices with i.i.d.\ entries that are mean zero and subgaussian with variance $1/m$.

\begin{definition}{\cite{KaneMN11}}
\textup{
Let $\mathcal{D}$ be a distribution over $\R^{m\times n}$. We say $\mathcal{D}$ has the {\em $(\eps,\delta,p)$-JL moment property} if for all $x\in\R^n$ of unit norm,
$$
\E_{\Pi\sim\mathcal{D}} |\|\Pi x\|^2 - 1|^p  < \eps^p \cdot \delta .
$$
}
\end{definition}

The following theorem follows from the proof of Lemma 8 in the full version of \cite{ClarksonW13}. We give a different proof here inspired by the proof of \cite[Theorem 9.9]{FoucartR13}, which is slightly shorter and more self-contained. A weaker version appears in \cite[Lemma 10]{Sarlos06}, where the size bound on $X$ is $(C d/\eps)^d$ for a constant $C\ge 1$ instead of simply $C^d$.

\begin{theorem}\TheoremName{nets}
Let $U\in\R^{n\times d}$ with orthonormal columns be arbitrary. Then there exists a set $X\subset\R^n$, $|X| \le 9^d$, each of norm at most $1$ such that
$$
\|(\Pi U)^T(\Pi U) - I\| \le 2 \cdot \sup_{x\in X}|\|\Pi x\|^2 - 1| 
$$
\end{theorem}
\begin{proof}
We will show that if $\sup_{x\in X}|\|\Pi x\|^2 - 1| < \eps/2$ then $\|(\Pi U)^T(\Pi U) - I\| < \eps$, where $\eps>0$ is some positive real. Define $A = (\Pi U)^T (\Pi U) - I$. Since $A$ is symmetric,
$$
\|A\| = \sup_{\|x\| = 1} |x^T A x| = \sup_{\|x\| = 1} |\inprod{Ax, x}|
$$
Let $T_{\gamma}$ be a finite $\gamma$-net of $\ell_2^d$, i.e. $T_{\gamma}\subset \ell_2^d$ and for every $x\in \R^d$ of unit norm there exists a $y\in T_{\gamma}$ such that $\|x - y\|_2 \le \gamma$. As we will see soon, there exists such a $T_{\gamma}$ of size at most $(1 + 2/\gamma)^d$. We will show that if $\Pi$ satisfies the JL condition on $T' = \{U y: y\in T_{1/4}\}$ with error $\eps/2$, then $\|A\| < \eps$; that is, $(1-\eps/2)\|x\|_2^2 \le \|\Pi x\|_2^2 \le (1+\eps/2)\|x\|_2^2$ for all $x\in T'$.

Let $x$ be a unit norm vector that achieves the $\sup$ above, i.e.\ $\|A\| = |\inprod{Ax, x}|$. Then, letting $y$ be the closest element of $T_{\gamma}$ to $x$,
\begin{align*}
\|A\| &= |\inprod{Ax, x}| \\
{}&= |\inprod{Ay, y} + \inprod{A(x + y), x-y}|\\
{} & \le \frac{\eps}2 + \|A\|\cdot \|x + y\| \cdot \|x - y\|\\
{} & \le \frac{\eps}2 + 2\gamma\|A\| .
\end{align*}
Rearranging gives $\|A\| \le \eps / (2(1 - 2\gamma))$, which is $\eps$ for $\gamma = 1/4$.

Now we must show that we can take $|T_\gamma| \le (1 + 2/\gamma)^d$. The following is a standard covering/packing argument for bounding metric entropy. Imagine packing as many radius-$(\gamma/2)$ $\ell_2$ balls as possible into $\R^d$, centered at points with at most unit norm and such that these balls do not intersect each other. Then these balls all fit into a radius-$(1+\gamma/2)$ $\ell_2$ ball centered at the origin, and thus the number of balls we have packed is at most the ratio of the volume of a $(1+\gamma/2)$ ball to the volume of a $\gamma/2$ ball, which is $((1+\gamma/2) / (\gamma/2))^d = (1 + 2/\gamma)^d$. Now, take those maximally packed radius-$(\gamma/2)$ balls and double each of their radii to be radius $\gamma$. Then every point in the unit ball is contained in at least one of these balls by the triangle inequality, which is exactly the property we wanted from $T_\gamma$ ($T_\gamma$ is just the centers of these balls). To see why every point is in at least one such ball, if some $x\in\R^d$ of unit norm is not contained in any doubled ball then a $\gamma/2$-ball about $x$ would be disjoint from our maximally packed $\gamma/2$ balls, a contradiction.
\end{proof}

\begin{lemma}\LemmaName{jl-to-ose}
If $\mathcal{D}$ satisfies the $(\eps, \delta, p)$-JL moment property, then $\mathcal{D}$ satisfies the $(2\eps, 9^d\delta, d, p)$-OSE moment property
\end{lemma}
\begin{proof}
By \Theorem{nets}, there exists a subset $X\subset\R^n$ of at most $9^d$ points such that
\begin{align}
\nonumber \E \|(\Pi U)^T(\Pi U) - I \|^p &\le 2^p \cdot \E \sup_{x\in X} | \|\Pi x\|^2 - 1|^p\\
\nonumber {}&\le 2^p \cdot \sum_{x\in X} \E | \|\Pi x\|^2 - 1|^p\\
\nonumber {}&\le 2^p \cdot 9^d \cdot \eps^p\cdot \delta \\
\nonumber {}&= (2\eps)^p \cdot 9^d \delta .
\end{align}
\end{proof}

It is known that if $\mathcal{D}$ is a distribution over $\R^{m\times n}$ with $m = \Omega(\log(1/\delta)/\eps^2)$ and for $\Pi\sim\mathcal{D}$, the entries of $\Pi$ are independent subgaussians with mean zero and variance $1/m$, then $\mathcal{D}$ has the $(\eps/2, \delta, \Theta(\log(1/\delta)))$-JL moment property \cite{KaneMN11}. Thus such a matrix has the $(\eps, \delta, d, \Theta(d + \log(1/\delta)))$-OSE moment property for $\delta < 2^{-d}$ by \Lemma{jl-to-ose}.

\subsection{Subsampled Randomized Hadamard Transform (SRHT)}\SectionName{srht}

Recall the SRHT is the $m\times n$ matrix $\Pi = (1/\sqrt{m})\cdot SHD$ for $n$ a power of $2$ where $D$ has diagonal entries $\alpha_1,\ldots,\alpha_n$ that are independent and uniform in $\{-1, 1\}$, $H$ is the unnormalized Hadamard transform with $H_{i,j} = (-1)^{\inprod{i,j}}$ (treating $i, j$ as elements of the vector space $\mathbb{F}_2^{\log_2 n}$), and $S$ is a sampling matrix. That is, the rows of $S$ are independent, and each row has a $1$ in a uniformly random location and zeroes elsewhere. A similar construction is where $S$ is an $n\times n$ diagonal matrix with $S_{i,i} = \eta_i$ being independent Bernoulli random variables each of expectation $m/n$ (so that, in expectation, $S$ selects $m$ rows from $HD$). We will here show the moment property for this latter variant since it makes the notation a tad cleaner, though the analysis we present holds essentially unmodified for the former variant as well. 

Our analysis below implies that the SRHT provides an $\eps$-subspace embedding for $d$-dimensional subspaces with failure probability $\delta$ for $m = O(\eps^{-2}(d + \log(1/(\eps\delta)))\log(d/\delta))$. This is an improvement over analyses we have found in previous works. The analysis in \cite{Tropp11} only considers constant $\eps$ and $\delta = O(1/d)$ and for these settings achieves $m = O((d + \log n)\log d)$, which is still slightly worse than our bound for this setting of $\eps, \delta$ (our bound removes the $\log n$ and achieves any $1/\mathop{poly}(d)$ failure probability with the same $m$). The analysis in \cite{LuDFU13} only allows failure probabilities greater than $n/e^d$. They show failure probability $\delta + n/e^d$ is achieved for $m = O(d\log(d/\delta)/\eps^2)$, which is also implied by our result if $m\le n$ (which is certainly the case in applications for the SRHT to be useful, since otherwise one could use the $n\times n$ identity matrix as a subspace embedding). The reason for these differences is that previous works operate by showing $HDU$ has small row norms with high probability over $D$; since there are $n$ rows, some logarithmic dependence on $n$ shows up in a union bound. After this conditioning, one then shows that $S$ works. Our analysis does not do any such conditioning at all. Interestingly, such a lossy conditioning approach was done even for the case $d=1$ \cite{AilonC09}. As we see below, these analyses can be improved (essentially the $\log n$ terms that appear from the conditioning approach can be very slightly improved to $\log m$).

Our main motivation in re-analyzing the SRHT was not to improve the bounds, but simply to clearly demonstrate that the SRHT satisfies the OSE moment property. The fact that our moment based analysis below (very slightly) improved $m$ was a fortunate accident. Before we present our proof of the OSE moment property for the SRHT, we state a theorem we will use. For a random matrix $M$, we henceforth use $\|M\|_p$ to denote $(\E\|M\|_{S_p}^p)^{1/p}$ where $\|M\|_{S_p}$ is the Schatten-$p$ norm, i.e.\ the $\ell_p$ norm of the singular values of $M$.

\begin{theorem}[Non-commutative Khintchine inequality {\cite{LP86,LPP91}}]\TheoremName{nck}
Let $X_1,\ldots,X_n$ be fixed real matrices and $\sigma_1,\ldots,\sigma_n$ be independent Rademachers. Then
$$
\forall p \ge 1,\ \|\sum_i \sigma_i X_i\|_p \lesssim \sqrt{p} \cdot \max\left\{\|(\sum_i X_i X_i^T)^{1/2}\|_{S_p}, \|(\sum_i X_i^T X_i)^{1/2}\|_{S_p} \right\} .
$$
\end{theorem}

We will also make use of the Hanson-Wright inequality.

\begin{theorem}[Hanson-Wright {\cite{HW71}}]\TheoremName{hw}
For $(\sigma_i)$ independent Rademachers and $A$ symmetric,
$$
\forall p \ge 1,\ \|\sigma^T A \sigma - \E \sigma^T A \sigma\|_p \lesssim \sqrt{p}\cdot\|A\|_F + p\cdot \|A\| .
$$
\end{theorem}

We now present our main analysis of this subsection.

\begin{theorem}\TheoremName{best-srht}
The SRHT satisfies the $(\eps, \delta, d, p)$-moment property for $p = \log(d/\delta)$ as long as $m\gtrsim \eps^{-2}(d\log(d/\delta) + \log(d/\delta)\log(m/\delta)) \simeq \eps^{-2}(d + \log(1/(\eps\delta))\log(d/\delta))$.
\end{theorem}
\begin{proof}
For a fixed $U\in\R^{n\times d}$ with orthonormal columns, we would like to bound 
$$
\E_{\alpha, \eta}\|\frac 1m(SHDU)^T(SHDU) - I\|^p.
$$
Since $p\ge \log d$ we have
\begin{equation}
\|\frac 1m(SHDU)^T(SHDU) - I\| \simeq \|\frac 1m(SHDU)^T(SHDU) - I\|_{S_p} \EquationName{holder}
\end{equation}
by H\"{o}lder's inequality. Also, let $z_1,\ldots,z_n$ be the rows of $HDU$, as column vectors, so that
\begin{equation}
\frac 1m(SHDU)^T (SHDU) = \frac 1m\sum_{i=1}^n \eta_i z_i z_i^T\EquationName{outer-products} .
\end{equation}
Note also $\sum_i z_i z_i^T = (HDU)^THDU = n\cdot I$ for any $D$, so the identity matrix is the expectation, over $\eta$, of the right hand side of \Equation{outer-products} for any $D$. Thus we are left wanting to bound
$$
\|\frac 1m\sum_i \eta_i z_i z_i^T - \E_{\eta'} \frac 1m\sum_i \eta_i' z_i z_i^T\|_p
$$
where the $\eta_i'$ are identically distributed as the $\eta_i$ but independent of them. Below we use $\|f(X)\|_{L^p(X)}$ to denote $(\E_X |f(X)|^p)^{1/p}$. Also we assume $p$ is an integer multiple of $4$, so that $\|A\|_{S_p}$ for real symmetric $A$ equals $(\Tr(A^p))^{1/p}$ and $\|A\|_{S_{p/2}} = (\Tr(A^{p/2}))^{2/p}$. Thus for $(\sigma_i)$ independent Rademachers,
\allowdisplaybreaks
\begin{align}
\|\frac 1m\sum_i \eta_i z_i z_i^T - & I\|_p = \|\frac 1m\sum_i \eta_i z_i z_i^T - \E_{\eta'} \frac 1m\sum_i \eta_i' z_i z_i^T\|_p\EquationName{beginning-sqrt}\\
\nonumber {}&= \|\|\frac 1m\sum_i \eta_i z_i z_i^T - \E_{\eta'} \frac 1m\sum_i \eta_i' z_i z_i^T\|_{L^p(\eta)}\|_{L^p(\alpha)}\\
\nonumber {}&\le \frac 1m \|\|\sum_i \eta_i z_i z_i^T - \sum_i \eta_i' z_i z_i^T\|_{L^p(\eta,\eta')}\|_{L^p(\alpha)}\text{ (Jensen's inequality)}\\
\nonumber {}&= \frac 1m\cdot \|\sum_i (\eta_i - \eta_i') z_i z_i^T\|_p\\
\nonumber {}&= \frac 1m\cdot \|\sum_i \sigma_i(\eta_i - \eta_i') z_i z_i^T\|_p\text{ (equal in distribution)}\\
\nonumber {}&\le \frac 2m \cdot \|\sum_i \sigma_i \eta_i z_i z_i^T\|_p\text{ (triangle inequality)}\\
\nonumber{}&\lesssim \frac{\sqrt{p}}m\cdot \|(\sum_i \eta_i \|z_i\|_2^2\cdot z_i z_i^T)^{1/2}\|_p\text{ (\Theorem{nck})}\\
\nonumber{}&\le \frac{\sqrt{p}}m\cdot \E\left((\max_i \eta_i \|z_i\|_2^p)\cdot \Tr((\sum_i \eta_i z_i z_i^T)^{p/2})\right)^{1/p} \text{ (}\|M\|_{S_p}^p = \Tr(M^p)\text{)}\\
\nonumber{}&\le \frac{\sqrt{p}}m\cdot \|\max_i \eta_i \|z_i\|_2^2\|_p^{1/2} \cdot (\E \Tr((\sum_i \eta_i z_i z_i^T)^{p/2})^2)^{1/2p}\text{ (Cauchy-Schwarz)}\\
&\le \frac{\sqrt{p}}m\cdot \|\max_i \eta_i \|z_i\|_2^2\|_p^{1/2} \cdot (d\cdot \E \Tr((\sum_i \eta_i z_i z_i^T)^p))^{1/2p}\text{ (Cauchy-Schwarz)} \EquationName{cs-twice}\\
\nonumber {}&\lesssim \frac{\sqrt{p}}m\cdot \|\max_i \eta_i \|z_i\|_2^2\|_p^{1/2} \cdot \|\sum_i \eta_i z_i z_i^T\|_p^{1/2} \text{ (since }d^{1/p} \le 2\text{)}\\
{}&\le \sqrt{\frac pm}\cdot \|\max_i \eta_i \|z_i\|_2^2\|_p^{1/2} \cdot (d^{1/p} + \|\frac 1m\sum_i \eta_i z_i z_i^T - I\|_p^{1/2})\text{ (triangle inequality)} \EquationName{end-sqrt}
\end{align}

\Equation{cs-twice} follows since if $\beta_i$ are the singular values of $M = \sum_i \eta_i z_iz_i^T$, then $\Tr(M^{p/2})^2 = (\sum_i \beta_i^{p/2})^2$, and the rank of $M$, and hence the number of summands $\beta_i$, is at most $d$. Letting $Q = \|\frac 1m\sum_i \eta_i z_i z_i^T - I\|_p^{1/2}$ and $R = \sqrt{p/m} \cdot \|\max_i  \eta_i \|z_i\|_2^2\|_p^{1/2}$, combining \Equation{beginning-sqrt} and \Equation{end-sqrt} 
$$
Q^2 \lesssim R + RQ
$$
implying that for some fixed constant $C > 0$, we have $Q^2 - CRQ - CR \le 0$. This implies that $Q$ is at most the larger root of the associated quadratic equation, i.e. $Q \lesssim \max\{\sqrt{R}, R\}$, or equivalently
\begin{equation}
\|\frac 1m\sum_i \eta_i z_i z_i^T - I\|_p \lesssim \max\{R, R^2\} \EquationName{almost-done}
\end{equation}

It only remains to bound $R$, which in turn amounts to bounding $\|\max_i \eta_i \|z_i\|_2^2\|_p^{1/2}$. Define $q = \max\{p, \log m\}$, and note $\|\cdot\|_p \le \|\cdot\|_q$. Then
\begin{align}
\nonumber \|\max_i \eta_i \|z_i\|_2^2\|_q &= \left(\E_{\alpha, \eta}\max_i \eta_i^q (\|z_i\|_2^2)^q\right)^{1/q}\\
\nonumber &\le \left(\E_{\alpha, \eta}\sum_i \eta_i^q (\|z_i\|_2^2)^q\right)^{1/q} \\
\nonumber &= \left(\sum_i \E_{\alpha, \eta} \eta_i^q (\|z_i\|_2^2)^q\right)^{1/q} \\
\nonumber &\le \left(n\cdot \max_i \E_{\alpha, \eta} \eta_i^q (\|z_i\|_2^2)^q\right)^{1/q} \\
\nonumber &= \left(n\cdot \max_i (\E_{\eta} \eta_i^q)\cdot (\E_\alpha (\|z_i\|_2^2)^q)\right)^{1/q} \text{ (}\alpha, \eta\text{ independent)}\\
\nonumber &= \left(m\cdot \max_i \E_\alpha (\|z_i\|_2^2)^q)\right)^{1/q}\\
\nonumber {}&\le 2\cdot \max_i \| \|z_i\|_2^2 \|_q\text{ (}m^{1/q} \le 2\text{ by choice of }q\text{)}\\
\nonumber {}&= 2\cdot \max_i \| \alpha^T \tilde{U}_i \tilde{U}_i^T \alpha\|_q \\
{}&= 2\cdot \max_i (d + \| \alpha^T \tilde{U}_i \tilde{U}_i^T \alpha - \E \alpha^T \tilde{U}_i \tilde{U}_i^T \alpha\|_q)\text{ (triangle inequality)} \EquationName{aboutto-hw}
\end{align}
where $\tilde{U}_i$ is the matrix with $(\tilde{U}_i)_{k,j} = H_{i,k}\cdot U_{k,j}$. Of particular importance for us is the identity $\tilde{U}_i^T \tilde{U}_i = I$. Then by \Equation{aboutto-hw} and \Theorem{hw},
\begin{align}
\nonumber \|\max_i \eta_i \|z_i\|_2^2\|_q &\lesssim d + \sqrt{q}\cdot \|\tilde{U}_i \tilde{U}_i^T\|_F + q\cdot \|\tilde{U}_i \tilde{U}_i^T\|\\
\nonumber {}&= d + \sqrt{qd} + q\\
\nonumber {}&\le \frac 32\cdot (d + q) \text{ (AM-GM inequality)}
\end{align}
so that
$$
R\lesssim \sqrt\frac pm\cdot \sqrt{d+q} ,
$$
which when combined with \Equation{almost-done} gives
$$
\|\frac 1m\sum_i \eta_i z_i z_i^T - I\|_p \lesssim \sqrt{\frac pm\cdot (d+q)} + \frac pm\cdot (d+q) .
$$
Thus the OSE moment property is satisfied by our choices of $m, p$ in the theorem statement.
\end{proof}

\subsection{Composing dimensionality reducing maps supporting AMM}\SectionName{composition}

As discussed in \Remark{efficiency}, to obtain both a good number of rows for $\Pi$ as well as fast multiplication for $\Pi A, \Pi B$, one may wish to set $\Pi$ as the composition $\Pi = \Pi_1 \Pi_2$, where $\Pi_1$ has the correct number $m_1 = O(k/\eps^2)$ of rows (e.g.\ a matrix of subgaussian entries), whereas $\Pi_2$ maps to a small but suboptimal number $m_2$ of rows (e.g.\ the SRHT) but supports fast embedding to compute $\Pi_2 A$. We show here that composing maps each supporting AMM yields a final map also giving AMM.

As discussed in \Corollary{bss-amm}, without loss of generality we can assume $A = B$. Also, as discussed in \Remark{equiv}, we can focus on achieving \Equation{op-error} where the number of rows of $\Pi$ should depend on the stable rank $\sr$ and not rank $\rank$ of $A$. The key is to note the following simple triangle inequality:
\begin{equation}
\|(\Pi A)^T(\Pi A) - A^T A\| \le \underbrace{\|(\Pi_1 \Pi_2 A)^T (\Pi_1 \Pi_2 A) - (\Pi_2 A)^T(\Pi_2 A)\|}_\alpha + \underbrace{\|(\Pi_2 A)^T(\Pi_2 A) - A^T A\|}_\beta .\EquationName{composition}
\end{equation}

The results of this work show that to achieve the desired $\beta \le \eps\|A\|^2$, it suffices that the number of rows of $\Pi_2$ need only depend on $\sr$ and not $\rank$, as desired. The trouble is that for $\alpha$, the number of rows of $\Pi_1$ will need to depend on the stable rank $\sr'$ of $\Pi_2 A$ and not $\sr$. Furthermore, the error will be $\alpha \le \eps\|\Pi_2 A\|^2$ and not $\alpha \le \eps\|A\|^2$. Thus, we must obtain good bounds on both $\sr'$ and $\|\Pi_2 A\|$. To achieve this, note
\begin{equation}
\|\Pi_2 A\| = \|(\Pi_2 A)^T(\Pi_2 A)\|^{1/2} = \|A\| \pm \|(\Pi_2 A)^T(\Pi_2 A) - A^T A\|^{1/2}  \EquationName{preserve-op}
\end{equation}
and
\begin{equation}
\|\Pi_2 A\|_F = \Tr((\Pi_2 A)^T (\Pi_2 A))^{1/2} = \|A\|_F \pm \|(\Pi_2 A)^T(\Pi_2 A) - A^T A\|_F\EquationName{preserve-frob}
\end{equation}

Thus, if we condition on $\beta = \|(\Pi_2 A)^T(\Pi_2 A) - A^T A\| \le \eps\|A\|^2$ (which we already discussed above), then indeed we have $\|\Pi_2 A\| = \Theta(\|A\|)$ by \Equation{preserve-op}. Also, \cite[Theorem 6.2]{KaneN14} implies $\|(\Pi_2 A)^T(\Pi_2 A) - A^T A\|_F \le \eps\|A\|_F^2$ with probability $1-\delta$ as long as $\Pi$ comes from a distribution satisfying the $(O(\eps), \delta, \ell)$-JL moment property for some $\ell\ge 2$ (which is just the $(O(\eps), \delta, 1, \ell)$-OSE moment property in the terminology of this work). If this holds, then  $\|\Pi_2 A\|_F = \Theta(\|A\|_F)$ by \Equation{preserve-frob}, and thus $\sr' = \Theta(\sr)$, as desired. Then overall, we have that the left hand side of \Equation{composition} is at most $\eps\cdot\|\Pi_2 A\|^2 + \eps\cdot\|A\|^2 = O(\eps)\cdot \|A\|^2$ as desired, in which both $\Pi_1$ and $\Pi_2$ need only provide AMM with error $\eps$ for matrices both of stable rank $O(\sr)$.

\begin{remark}
\textup{
An even slicker argument that works in the case when $\Pi_1, \Pi_2$ are both drawn from distributions satisfying the $(\eps, \delta, k, \ell)$-OSE moment property is to observe that the distribution of the product $\Pi_1 \Pi_2$ itself satisfies the OSE moment property. Indeed, letting $\|Z\|_p$ denote $(\E|Z|^p)^{1/p}$ for a scalar random variable $Z$, and letting $U\in\R^{n\times k}$ denote a matrix with orthonormal columns,
\begin{align*}
\| \|(\Pi_1 \Pi_2 U)^T \Pi_1 \Pi_2 U - I\| \|_\ell &< \eps \delta^{1/\ell} \| \|\Pi_2 U\|^2 \|_\ell \text{ (\Lemma{matmult})}\\
{}&= \eps \delta^{1/\ell} \| \|(\Pi_2 U)^T \Pi_2 U\| \|_\ell\\
{}&\le \eps \delta^{1/\ell}(1 +  \| \|(\Pi_2 U)^T \Pi_2 U - I\| \|_\ell)\text{ (triangle inequality)}\\
{}&\le \eps \delta^{1/\ell}(1 + \eps \delta^{1/\ell})
\end{align*}
In the first line we used that when $A = B$ in \Lemma{matmult}, $\Pi_1$ need only satisfy the OSE moment property with parameter $k$ instead of $2k$ (since then the span of the columns of both $A$ and $B$ has dimension at most $k$). Thus the distribution of the product $\Pi_1 \Pi_2$ satisfies the $(O(\eps), O(\delta), k, \ell)$-OSE moment property.
}
\end{remark}

\end{document}